\documentclass[a4paper,UKenglish,cleveref, autoref, thm-restate]{lipics-v2021}

\usepackage{epsfig}
\usepackage{epstopdf}
\usepackage{graphicx}
\DeclareGraphicsExtensions{.pdf,.png,.jpg,.eps}
\usepackage{latexsym}
\usepackage{amsfonts}
\usepackage{bm}
\usepackage{algorithm} 
\usepackage{algpseudocode}
\usepackage{multirow}
\usepackage[noadjust]{cite}
\usepackage{url}
\usepackage{pifont}
\usepackage{hyperref}
\newcommand{\T}{\mathcal{T}}

\usepackage{amsmath,amssymb,amstext,amsthm}
\usepackage{enumitem}
\usepackage{verbatim}

\newcommand\ceil[1]{\lceil #1 \rceil}
\newcommand\floor[1]{\lfloor #1 \rfloor}

\newcommand{\RMQ}{{\textsf{RMQ}}}
\newcommand{\access}{{\textsf{access}}}
\newcommand{\lca}{{\textsf{lca}}}
\newcommand{\depth}{{\textsf{depth}}}

\newcommand{\preorder}{{\textsf{preorder}}}
\newcommand{\inorder}{{\textsf{inorder}}}

\title{Succinct Representation of Search Trees on Trees} 
 \author{Seungbum Jo}{Chungnam National University, South Korea}{sbjo@cnu.ac.kr}{https://orcid.org/0000-0002-8644-3691}{This work was supported by the National Research Foundation of Korea(NRF) grant funded by the Korea government(MSIT) (RS-2025-23963814)}

\author{Nodari Sitchinava}{University of Hawaii at Manoa, United States}{nodari@hawaii.edu}{https://orcid.org/0000-0001-8876-4846}{Supported by NSF grant 2432018.}

\authorrunning{S. Jo and N. Sitchinava} 

\Copyright{Seungbum Jo and Nodari Sitchinava} 

\ccsdesc[500]{Theory of computation~Data structures design and analysis} 

\keywords{Search trees on trees, Encoding data structures, Binary search trees} 

\relatedversion{} 

\EventEditors{John Q. Open and Joan R. Access}
\EventNoEds{2}
\EventLongTitle{42nd Conference on Very Important Topics (CVIT 2016)}
\EventShortTitle{CVIT 2016}
\EventAcronym{CVIT}
\EventYear{2016}
\EventDate{December 24--27, 2016}
\EventLocation{Little Whinging, United Kingdom}
\EventLogo{}
\SeriesVolume{42}
\ArticleNo{23}
\nolinenumbers
\begin{document}

\maketitle
\begin{abstract}
A search tree on trees (STT) is a data structure for performing a search for a target vertex in a reference tree. A standard binary search tree is a special case of an STT, where the reference tree is a path of totally ordered elements.

In this paper, we study the problem of succinct representation of STTs.
We consider two cases: (1) general search trees on trees, and  
(2) Steiner-closed search trees on trees [Bose et al. TALG 2023].  
For both cases, we present representations that can be constructed in polynomial time  
and achieve optimal space up to the lower-order additive terms.  

We also present data structures for supporting fast traversals of both general and Steiner-closed STTs.
For general STTs our data structure still takes optimal space up to lower-order additive terms.
\end{abstract}

\newpage
\section{Introduction}
\emph{Search on trees} (and more generally \emph{search on graphs}) has been extensively studied under various names in computer science and discrete mathematics, such as elimination trees~\cite{elimination-trees-1,elimination-trees-2,elimination-trees-3,elimination-trees-4}, ordered-colorings~\cite{ordered-colourings}, edge rankings~\cite{d08}, vertex rankings~\cite{vertex-rankings-1,vertex-rankings-2,vertex-rankings-3}, and tubings~\cite{tubings-1,tubings-2}. 
It is a generalization of binary search on a totally ordered set, which can be viewed as a search on a graph that is a path. (See Section~\ref{sec:related} for related work on this topic.)

Recently, Bose et al.~\cite{DBLP:journals/talg/BoseCIKL23} took a data structuring approach to this problem and presented the \emph{search tree on tree (STT)} data structure, which is a generalization of the classical \emph{binary search tree (BST)}. 
Given an (unrooted) tree $T$, referred to as the \emph{reference tree}, the root of the STT $\mathcal{T}$ on $T$ corresponds to some vertex $x \in T$ and its subtrees are STTs on each of the connected components of $T \setminus x$. 
Therefore, BST $\T$ is a special case of STTs, where $\T$ is built on the reference tree $T$ that is a path -- the Hasse diagram of the totally ordered set $P$. 

STT data structures (and their algorithmic analogues for search on trees and graphs) have been of increasing interest recently (see Section~\ref{sec:related} for references). In this paper, we ask the next natural question in the study of STTs: how efficiently can we encode STTs?

\subsection{Related work}\label{sec:related}
Binary search on a totally ordered set is a classical algorithm.
BSTs were considered independently by several researchers in the 1950s~\cite[Sec. 6.2.2]{knuthart3sorting}, while the first balanced dynamic BST with worst-case $O(\log n)$ height -- the AVL tree -- was invented in 1962~\cite{avl-tree}.

The question of generalizing the idea of binary search from totally ordered sets to partially ordered sets (\emph{posets}) was initiated by Linial and Saks~\cite{ls85}. They showed that $\Theta(\log n)$ queries are necessary and sufficient to find a target element in the worst case. 
Their result can be viewed as a search for a target vertex in a directed acyclic graph (DAG) -- the Hasse diagram that represents the partial order relation between the elements of the poset. In this view, the search algorithm queries an oracle to determine which vertex of a directed edge $(u,v)$ is closer to the target vertex $t$ and proceeds in the corresponding subgraph (the subgraph induced by the set of vertices reachable from the vertex returned by the oracle). Hence, the result of Linial and Saks can be viewed as building a search tree on a DAG, with queries being performed on the edges. Alternatively, one can consider the search on DAGs where the queries are performed on the vertices. In this setting, a query on vertex $u$ returns the neighbor $v$ of $u$ that is closer to target $t$ than $u$. 

Linial and Saks~\cite{ls85} showed that, just like in BSTs, the minimum height of a search tree on DAGs is $\Theta(\log n)$ in the worst case. However, unlike BSTs, the problem of finding minimum-height search trees on DAGs is known to be NP-hard~\cite{cdkl04,d08}. Therefore, the bulk of existing work is focused on search trees on trees rather than DAGs. The minimum-height STTs can be computed in linear time~\cite{s89,ly01,mow08,op06}, and this restriction for search on trees rather than DAGs is sufficient for many practical applications, such as file system synchronization~\cite{bfn99,mow08}, parallel database queries~\cite{mui01}, software testing~\cite{bfn99,mow08}, asymmetric communication protocols~\cite{lm11},  assembly planning~\cite{irv88}, and mobile agent computing~\cite{bfkr21}. 

Search on trees and more generally search on graphs has been studied for decades under various names, such as elimination trees~\cite{elimination-trees-1,elimination-trees-2,elimination-trees-3,elimination-trees-4}, ordered-colorings~\cite{ordered-colourings},
and tubings~\cite{tubings-1,tubings-2} among others. 
The minimal height of a search tree on a graph $G$ is also known as the \emph{treedepth} of $G$~\cite[Ch. 6]{sparsity-book-12} and computing the search tree on trees  is equivalent to the problems of \emph{edge ranking}~\cite{d08} (when queries are on the edges) and \emph{vertex ranking}~\cite{vertex-rankings-1,vertex-rankings-2,vertex-rankings-3} (when queries are on the vertices). 
Because of different names in various research communities, several results have been rediscovered multiple times.  For example, linear time algorithms for vertex rankings was discovered by Sch\"affer~\cite{s89} and for edge rankings by Lam and Yue~\cite{ly01} and rediscovered in the context of STTs by Mozes et al.~\cite{mow08} and Onak and Parys~\cite{op06}, respectively.

The question of minimizing the height of a search tree (worst-case optimality) is rather limited. Alternatively, one might ask for a search tree that minimizes the total query time across a given distribution of queries (also known as \emph{optimal BST}~\cite{knuth-opt-bst71}). In case of classical BSTs, the shape of such optimal BST can be found in $O(n^2)$ time, using the classic dynamic programming algorithm by Knuth~\cite{knuth-opt-bst71}. 
On the other hand, finding the optimal search tree on graphs with edge queries is NP-hard even for trees of diameter at most \(4\)~\cite{cjlm11}, although a constant-factor approximation can be computed in linear time~\cite{lm11}.

\subparagraph*{Adaptive STTs.}
Recently, Bose et al.~\cite{DBLP:journals/talg/BoseCIKL23} considered the adaptive version of the STTs (with vertex queries). They proposed an extension of the BST rotation operation to the STT setting. They presented an analog of the Tango trees in the STT setting and showed that it is $O(\log\log n)$-competitive with an optimal STT on a distribution of queries that are not known a priori, i.e., that are produced online. Crucial to achieving this result is the concept of a restricted class of STTs, called \emph{Steiner-closed} STTs, which they showed can be constructed from general STTs with a factor of at most 2 increase in the depth of every node. 

Berendsohn and Kozma~\cite{splayTT} presented \emph{Splay tree on trees (SplayTT)} -- a generalization of splay trees~\cite{DBLP:journals/jcss/SleatorT83} to the STT setting. They showed that SplayTTs achieve static optimality, i.e., are constant-factor competitive with any optimal static STTs on the online sequence of queries.  
They also conjectured that SplayTTs are \emph{dynamically optimal}, i.e., are constant-factor competitive with any optimal \emph{dynamic} STT on the online sequence of queries, similar to how splay trees are conjectured to be dynamically optimal in the BST setting.

\subparagraph*{Succinct tree representations.}
Search trees are widely used as data structures that organize and enable efficient access to data, 
as seen in applications such as text indexing~\cite{navarro2007compressed} and information retrieval~\cite{muthukrishnan2002efficient}. 
One way to access indexed data fast is to store the data within the nodes of the search tree. 
However, as data grows, the overhead associated with copying and storing it within the nodes of the tree becomes significant. 
Instead, one may want to store the shape of the search tree and access the data indexed by the tree directly from the original storage medium. This motivates space-efficient representations of tree structures that support efficient navigation while using minimal additional space.

Considering only the tree structure and ignoring the labels, a search tree can be represented as an ordered tree. A naive approach to storing an ordered tree with $n$ nodes is to use $O(n)$ pointers, which requires $\Omega(n \log n)$ bits\footnote{We assume a word-RAM model with a word size of $\Theta(\log n)$ bits. All our logarithms are of base $2$.}.
However, the information-theoretic space lower bound for storing such trees is $2n-O(\log n)$ bits. 
Thus, research in this area mainly focuses on reducing the gap between the naive representation and the information-theoretic lower bound.
Since Jacobson's LOUDS representation~\cite{DBLP:conf/focs/Jacobson89}, various \textit{succinct} representations (i.e., representations of size $2n+o(n)$ bits) have been proposed for ordered and unlabeled trees, enabling efficient support for tree navigation queries (see \cite{navarro2016compact} for earlier results on succinct tree representations). 
One of the most recent succinct representations, proposed by Navarro and Sadakane~\cite{DBLP:journals/talg/NavarroS14}, supports a wide range of tree navigation queries in $O(1)$ time (see Table 1 in \cite{DBLP:journals/talg/NavarroS14}).

Binary search trees (BSTs) can be considered as ordered and labeled trees where each node is labeled with a key from the totally ordered set, and all keys in the left (resp. right) subtree are smaller (resp. larger) than the key at the node. The data structure of Navarro and Sadakane~\cite{DBLP:journals/talg/NavarroS14} provides a ($2n+o(n)$)-bit representation of a static BST, supporting traversal in $O(1)$ time per node, as the label of any node can be obtained from its inorder number, while the keys are stored in a separate array, not counted in the space of BST. 

\subsection{Our results}
In this paper, we study the problem of storing STTs using small space.
Specifically, given a reference tree $T$ with $n$ vertices, which is accessible during decoding, our goal is to obtain a space-efficient representation of the STT $\T$ on $T$. Since we assume that $T$ is accessible during decoding, we do not account for the space to store it.
We show that we can significantly improve on the naive $O(n \log n)$-bit encoding of $\T$:

\begin{enumerate}
	\item In Section~\ref{sec:general}, we present an $(\ceil{n \log \ell}+ 2n)$-bit representation for $\T$, where $\ell$ denotes the number of leaves in $T$. 
	Moreover, in \Cref{sec:lb}, we show that there exists $T$ that requires at least $n \log {\ell}-O(\ell \log {(n/\ell)})$ bits to encode an STT on $T$, i.e., our representation is optimal up to $O(n)$ additive terms in general, and up to the \emph{lower-order} additive terms when $\ell = \omega(1)$.
	
	\item In Section~\ref{sec:steiner} we show that if $\T$ is a Steiner-closed STT on $T$ (see Section~\ref{sec:prelim} for the formal definition), there exists a $(2n + \ceil{\log n})$-bit representation of $\T$. 
    Furthermore, \Cref{sec:lb} we prove that any Steiner-closed search tree on $T$ requires at least $2n - O(\log n)$ bits, i.e., our representation is optimal up to lower-order additive terms. In particular, the upper and lower bounds coincide with those for BSTs on arrays of size $n$, up to lower-order additive terms. Since every STT can be transformed into a Steiner-closed STT at a cost of only doubling the height~\cite{DBLP:journals/talg/BoseCIKL23}, i.e., without asymptotically affecting search time, the more efficient representation of Steiner-closed STTs is of particular interest.
\end{enumerate}

Both of the above representations can be constructed in polynomial time. However, they do not support traversal of $\T$ efficiently. 
Instead, in \Cref{sec:access}, we present additional representations of general and Steiner-closed STTs that support traversal from the current node to its $i$-th child, where the current node is identified by its preorder number, in the following time bounds:

\begin{enumerate}
	\item An $(n \log \ell + O(n) + o(n \log  \ell))$-bit representation that supports access to the $i$-th child in $O(\log \log \ell)$ time, and
	\item An $O(n)$-bit representation that supports access to the $i$-th child in $O(1)$ time when $\T$ is Steiner-closed.
\end{enumerate}

Note that, while the first of these two representations achieves optimal space up to lower order terms,  
the second one is optimal only asymptotically, because only $2n-O(\log n)$ bits are necessary to represent Steiner-closed STTs. 

The rest of the paper is organized as follows.
Section~\ref{sec:prelim} introduces the necessary notations.
Section~\ref{sec:summary} presents a high-level overview of the techniques used to obtain the results. 
Sections~\ref{sec:general} and~\ref{sec:steiner} present the succinct representations of general STTs and Steiner-closed STTs, respectively.

\section{Preliminaries}\label{sec:prelim}

Table~\ref{tab:symbols} summarizes the notations used throughout the paper.

\begin{table}[tp]
	\centering
	\begin{tabular}{l p{5cm} l p{5cm}}
		\hline
		\textbf{Symbol} & \textbf{Description} & \textbf{Symbol} & \textbf{Description} \\
		\hline
		\hline
		$T(v)$ & The subtree of $T$ rooted at $v$. 
		& $CH_G(S)$ & Convex hull of the vertex set $S$ in $G$. \\
		$G[S]$ & Induced subgraph of $G$ by the vertex set $S$. 
		& $\lca_T(u, v)$ & The lowest common ancestor (LCA) of $u$ and $v$ in the rooted tree $T$. \\
		$P_G(u,v)$ & The path from $u$ to $v$ in $G$. 
		& $\depth_T(v)$ & The depth of $v$ in the rooted tree $T$. \\
		$\preorder_T(v)$ & The preorder number of $v$ in the ordered tree $T$. 
		& $\inorder_T(v)$ & The inorder number of $v$ in the binary tree $T$. \\
		\hline
		\hline
	\end{tabular}
	\caption{Summary of the main notations used in this article}
	\label{tab:symbols}
\end{table}

\subparagraph*{Search trees on trees.} 
Given a labeled tree $T = (V, E)$ with $n$ vertices, a \textit{search tree $\T$ on $T$} is defined as a rooted tree on $T$ satisfying the following conditions:

\begin{itemize}
	\item The root node $r$ of $\T$ corresponds to an arbitrary vertex of $T$.
	\item The rooted subtrees of $\T \setminus r$ are search trees on the connected components of $T \setminus r$.
\end{itemize}

We assume the labels of $T$ are totally ordered and unique, and given the label of a node $u$, we are able to access $u$ and all relevant information stored at $u$ in $O(1)$ time. For example, the label can be the memory address of where $u$ is stored. 

By definition, $T$ and its search tree $\T$ share the same set of nodes $V$. 
In the rest of the paper, to avoid ambiguity, we use the term \textit{vertex} for elements of $T$ and \textit{node} for elements of $\T$.
Furthermore, we refer to the vertex (resp., node) $u$ as the vertex (resp., node) with label $u$.
In general, $T$ is an unrooted tree, and $\T$ is defined as an unordered rooted tree.
In this paper, however, we assume that $T$ is a rooted and ordered tree (the root of $T$ and the ordering of 
siblings are defined in later sections). 
Moreover, we treat $\T$ as an ordered rooted tree by defining, for any non-leaf node $d$, the sibling order according to the \textit{smallest} preorder numbers among the vertices in the corresponding connected components of $T \setminus \{d\}$.\footnote{In general, the preorder traversal of an ordered tree visits the root $r$ first and then recursively performs preorder traversal on the subtrees rooted at the children of $r$ from left to right.} For example, in $T$ and $\T$ shown in Figure~\ref{figure:tit}(a) and (b), respectively, the node $i$ is the leftmost sibling among the children of $d$, because the connected component of $T$ that corresponds to the subtree of $\T$ rooted at $i$ contains vertex $a$, which has the smallest preorder number among all vertices in $T$.

\begin{figure}[t]
	\begin{center}
		\includegraphics[scale=0.7]{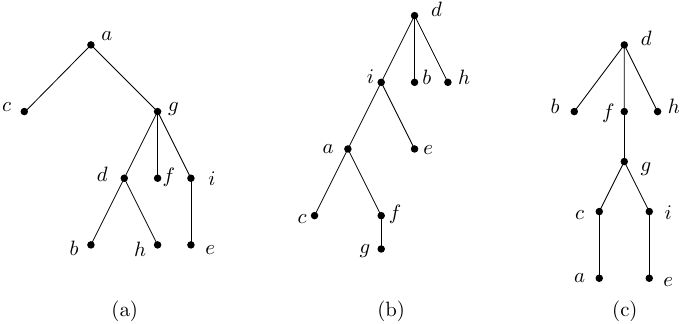}
	\end{center}
	\caption{(a) Reference tree $T$, (b) search tree $\T$ on $T$, and (c) Steiner-closed search tree $\T$ on $T$.}
	\label{figure:tit}
\end{figure}

As a special case of a search tree on $T$, Bose et al.~\cite{DBLP:journals/talg/BoseCIKL23} introduced the notion of a \textit{Steiner-closed search tree} on $T$, defined as follows.
Given a subset of vertices $S \subseteq V$ of a graph $G$, the \textit{convex hull} $CH_G(S)$ is the set of all vertices 
lying on the paths between any pair of vertices $u,v \in S$ in $G$. 
The set $S$ is called \textit{Steiner-closed with respect to $G$} if and only if every vertex in $CH_G(S) \setminus S$ has degree exactly two in the subgraph of $G$ induced by $CH_G(S)$.
A search tree $\T$ on $T$ is Steiner-closed if and only if, for every node $v$, the set of nodes along the path from the root to $v$ in $\T$ forms a Steiner-closed set with respect to $T$. 
For example, the tree in Figure~\ref{figure:tit}(b) is a search tree on the tree shown in Figure~\ref{figure:tit}(a),  
but it is \textit{not} a Steiner-closed STT, because the set of nodes on the path from $d$ to $a$ in $\T$, i.e., $S = \{a, d, i\}$, 
does not form a Steiner-closed set with respect to $T$, as it omits the node $g$, which has degree $3$ in the subgraph of $T$ induced by $CH_T(S) = \{a, d, g, i\}$.

\subparagraph*{Supporting traversal of $\T$.} 
To support the traversal of $\T$, we identify the current node by its preorder number $k$ in $\T$. The search starts at the root, for which $k=1$. Given $k$, our $\access(k)$ operation returns the label $u$ of the corresponding node. The oracle $\mathcal{O}$, given the search key and $u$, returns the index $i$ of the child from which the search should proceed.
For example, when $T$ is a path and $\T$ is a BST, the oracle query is the comparison query that returns the left child ($i=0$) or the right child ($i=1$) of $u$ from which the search should proceed.
Using the constant-time navigation operations of Navarro and Sadakane~\cite{DBLP:journals/talg/NavarroS14}, we then obtain the preorder number $k'$ of the $i$-th child of the current node. 
We continue the search with $k'$ as the preorder number of the current node. 
Thus, the remaining problem is to support $\access(k)$, which maps the preorder number of a node in $\T$ to its label.

\section{Technical Overview}\label{sec:summary}
In this section we present a high-level overview of our techniques.

\subparagraph*{Succinct representations of general STTs.}
Given a reference tree $T$ with $\ell$ leaves (recall that we treat both $T$ and $\T$ as rooted and ordered trees), we start by decomposing $T$ into $\ell$ ancestor-to-leaf disjoint paths $P_1, \dots, P_{\ell}$.
This can be achieved via a preorder traversal of $T$.

\begin{figure}[t]
	\begin{center}
		\includegraphics[scale=0.7]{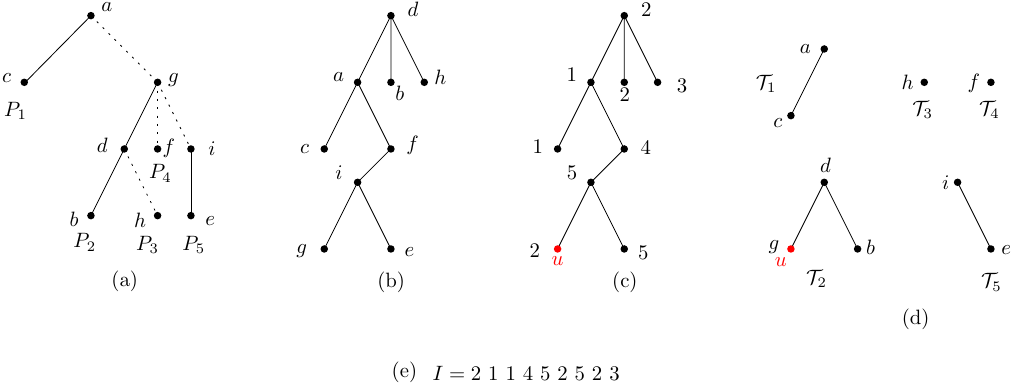}
	\end{center}
	\caption{(a) A tree $T$ with paths $P_1, \dots, P_5$ (dotted lines indicate edges that are not part of any path in $P$), (b) a search tree $\T$ on $T$, (c) $\T'$, (d) the tree structures of $\T_1, \dots, \T_5$, (e) the array $I$, which stores the labels of $\T'$ in preorder (this is used in the representation of \Cref{thm:general_upper}).
		}
	\label{figure:general}
\end{figure}

Based on these paths, our representation of $\T$ is a labeled tree $\T'$ that has the same structure as the STT 
$\T$, where each node of $\T'$ is labeled by the index of the path that contains its corresponding vertex in $T$.
Since there are at most $\ell$ distinct path labels, $\T'$ can be represented using $\ceil{n \log \ell}+2n$ bits. 
Moreover, since $P_1, \dots, P_{\ell}$ are ancestor-to-leaf paths, we can construct $\T_1, \dots, \T_{\ell}$ (for each $i \in [\ell]$, $\T_i$ is constructed from the nodes of $\T$ whose corresponding vertices lie in $P_i$) directly from $T$ and $\T'$ without storing them explicitly. 
Hence, the size of the representation is $\ceil{n \log \ell}+2n$ bits in total. 

Using this representation, we can fully reconstruct $\T$: $\T'$ provides the tree structure of $\T$, and from the label of each node in $\T'$, we can identify its corresponding node in the binary trees $\T_1, \dots, \T_{\ell}$, whose labels can in turn be computed from $T$.
For details of the representation and the reconstruction procedure, see Section~\ref{sec:general}.

Finally, we prove the space lower bound for storing $\T$ by showing that a \textit{comb-shaped tree} (see Figure~\ref{figure:general_lower}) requires $n \log \ell - O(\ell \log (n/\ell))$ bits to be represented by counting the number of possible STTs on such a tree. 
Details are presented in \Cref{sec:lb}.

\begin{figure}[t]
	\begin{center}
		\includegraphics[scale=0.5]{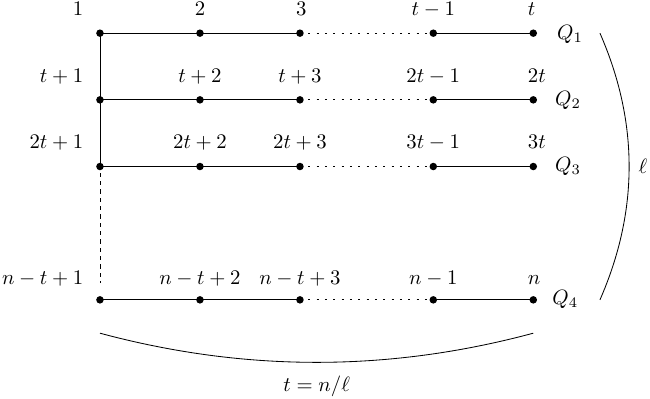}
	\end{center}
	\caption{A comb-shaped tree used in proving the space lower bound for general STTs.}
	\label{figure:general_lower}
\end{figure}

\subparagraph*{Succinct representations of Steiner-closed STTs.}
When $\T$ is a Steiner-closed STT, the Steiner-closed property of $\T$ allows us to find a specific decomposition of $T$ into ancestor-to-leaf disjoint paths $P_1, \dots, P_{\ell}$, such that  each set of vertices $V_i$ that forms path $P_i$ also forms a connected component $\T[V_i]$ in $\T$ (such decomposition might not exist in general STTs). 
The fact that $\T[V_i]$ is a connected component in $\T$ allows us to reconstruct $\T$ only from the representations of $\T_1,\dots, \T_\ell$ (the representations of the paths $P_i$) without the need for $\T'$. 
Instead it is sufficient to maintain a data structure that can answer \textit{Range Minimum Queries (RMQ)} on a sequence of length $\ell$ that captures the information about how the paths $P_1, \dots, P_{\ell}$ are formed from $T$.
Because the paths $P_1, \dots, P_{\ell}$ and their corresponding binary search trees $\T_1,\dots, \T_\ell$ form connected subgraphs in both $T$ and $\T'$, respectively, the answers to the RMQ queries provide a guide for merging $\T_1,\dots, \T_\ell$ into a single tree $\T$. 

We store the tree structures of $\T_1,\dots, \T_\ell$ using only $2n - 2\ell$ bits in total, while the RMQ data structure requires only $2\ell$ bits~\cite{bender2000lca,DBLP:journals/cacm/Vuillemin80}.
As a result, we obtain a $(2n + \ceil{\log n})$-bit encoding of $\T$ (where the additional $\log n$ bits are used to store the root of $\T$). 
For details of the representation and the reconstruction procedure, see Section~\ref{sec:steiner}.

Finally, in \Cref{sec:lb}, we show that this is tight up to lower-order terms by observing that every STT on a path is Steiner-closed, and that STTs on a path of size $n$ are in one-to-one correspondence with BSTs on an ordered set of the same size.

\subparagraph*{Supporting efficient traversals via $\access(k)$ queries.}
Consider the following procedure for performing the search for some key $\kappa$  on STT $\T$. 
We identify the current node by its preorder number $k$ in $\T$, starting from the root. 
Given $k$, $\access(k)$ returns the label $u$ of the current node. 
The query oracle takes $u$ together with the target key $\kappa$ and determines the index $i$ of the child of 
$u$ in $\T$ from which the search should proceed. 
Thus, we would like to support efficient traversal from the current node to its \(i\)-th child, where the current node is identified by its preorder number.

The succinct representation of Navarro and Sadakane~\cite{DBLP:journals/talg/NavarroS14} supports a number of operations on static trees in $O(1)$ time, including the operation of computing the preorder number $k'$ of the $i$-th child of the node with preorder number $k$. 
Therefore, after each traversal step, it remains to obtain the label of the node whose preorder number in $\T$ is $k'$. We call such a query $\access(k)$.

Clearly, storing the mapping from each preorder number to the corresponding label naively would require $\Omega(n\log n)$ space. Instead, we describe space-efficient data structures that support answering  $\access(k)$ queries on general and Steiner-closed STTs.
When $\T$ is a general STT, the data structure is relatively simple: it stores a succinct representation of $\T$ together with $O(n)$-bit auxiliary structures that support efficient navigation on the trees in the representation (namely, $\T'$ and $\T_1, \dots, \T_{\ell}$).

However, there are additional challenges in obtaining an $O(n)$-bit data structure that supports $\access{}$ efficiently for a Steiner-closed STT $\T$.
The main challenge lies in computing, for a query node, the index of the path that contains its corresponding vertex efficiently. In the general STT case, this index can be obtained directly from $\T'$, but $\T'$ cannot be stored using only $O(n)$ bits.

Instead, we introduce the \textit{extended search tree} $\T^e$ of $\T$, which is constructed from $\T$ by adding at most $n$ \textit{copied nodes}.
In addition to all core properties of $\T$, 
$\T^e$ has one extra property that enables us to compute the path indices efficiently: 
the relative order of (the first) occurrences of paths $P_1, \dots, P_\ell$ in the preorder traversal of $T$ is the same as the relative order of (the first) occurrences of paths $P_1, \dots, P_\ell$ in the preorder traversal of $\T^e$. 
This property does not hold in Steiner-closed STTs in general.

After determining the path index $j$ containing the corresponding vertex of the query node, we compute the two endpoint vertices of $P_j$ in $O(1)$ time using $O(n)$ auxiliary bit sequences.
Finally, we compute the label of the query node in $O(1)$ time by
(1) locating the query node within the tree $\T_j$ using $\T^e$, and
(2) computing the label of the node in $\T_j$ from $P_j$.
Details of the data structures are presented in \Cref{sec:access}.

\section{Succinct Representation of General Search Trees on Trees}\label{sec:general}
Consider a tree $T = (V, E)$ with $n$ vertices and $\ell$ leaves.
In this section we treat $T$ as a rooted and ordered tree by assuming that (1) the root $r$ of $T$ is the vertex with the smallest label, and (2) for each non-leaf vertex $u \in T$, the children of $u$ are ordered from left to right according to their labels.
Since the labels are totally ordered, we can compute the preorder number of any vertex in $T$ without using extra space.

The main idea of the representation of the STT $\T$ on $T$ is to first decompose $T$ into $\ell$ distinct paths, 
where each path connects a leaf vertex to one of its ancestors in $T$.
For each path, we construct an STT on that path from the structure of $\T$. 
Since each such search tree can be viewed as a standard BST on an array (see Lemma~\ref{lem:path}), 
we exploit key structural properties of $\T$ from Lemma~\ref{lem:tit} to design a representation.
In particular, we encode $\T$ by storing the tree structure of $\T$ (together with the index of the path containing each node) and the search trees on the paths.
This representation uses $\ceil{n \log \ell}+ 2n$ bits in total, 
and 
Theorem~\ref{thm:general_lower} shows that this representation is optimal up to an additive $O(n)$ term.

We now proceed to describe the representation in detail.
We first prove a lemma showing that if $T$ is a path, then any search tree on $T$ can be treated as a standard BST.

\begin{lemma}\label{lem:path}
	There exists a $(2n-2)$-bit representation of any search tree $\T$ on a path $T$.
\end{lemma}
\begin{proof}
For $n=1$, $T$ and $\mathcal T$ are identical, so there is no need to maintain $\mathcal T$ separately.
For $n>1$, root $T$ at an endpoint and order its vertices by depth; then any STT $\mathcal T$ on $T$ is a BST whose inorder order coincides with the preorder order of $T$; a single child is treated as the left or right child according to this order.
Applying the first-child--next-sibling representation and BP (balanced parentheses sequence)~\cite{munro2001succinct} encoding gives $2n+2$ bits, and omitting the first two opening and last two closing parentheses yields $2n-2$ bits.
\end{proof}

Next, we describe how to decompose $T$ into $\ell$ disjoint paths $P_1, \dots, P_{\ell}$.
For $i \in [\ell]$, let $l_i$ be the $i$-th leaf vertex of $T$ (leaves are ordered by their preorder numbers in $T$).
First, define $P_1 = P_T(r, l_1)$, and for $i \in\{2, \dots, \ell\}$, define
$P_i$ as the induced subgraph of $T$ consisting of the vertices in $P_T(r, l_i) \setminus \left(\cup_{j=1}^{i-1} P_j\right)$.
By construction, the sets of vertices on the paths $P_1, \dots, P_{\ell}$ form a partition of $V$.
The following lemma characterizes how the vertices in each path appear in $\T$.

\begin{restatable}{lemma}{propertiesgeneral}\label{lem:tit}
	For any vertices $u, v \in P_i$, the following statements hold: 
	\begin{enumerate}
		\item $\T(u)$ has at most two child subtrees that contain at least one node in $P_i$.
		\item Let $\T(u_l)$ and $\T(u_r)$ be the two non-empty child subtrees of $\T(u)$ 
		each containing at least one node in $P_i$.  
		If $u_l$ is a left sibling of $u_r$, then $\T(u_l)$ (resp., $\T(u_r)$) contains all nodes $u' \in P_i \cap \T(u)$ with $\depth_T(u') < \depth_T(u)$ (resp., $\depth_T(u') > \depth_T(u)$).
		\item $\lca_{\T}( u, v) \in P_i$.
	\end{enumerate} 
\end{restatable}
\begin{proof}
	By the construction of $P_i$, the vertices in $P_i$ have distinct depths in $T$, which directly implies statement $1$ from the definition of a search tree on $T$. 
	Statement $2$ also follows from the definition of the ordering among siblings in $\T$.
	For any vertex $x \in P_i$, there are at most two connected components in $T \setminus \{x\}$ that contain at least one vertex in $P_i$: $C^u$, the connected component containing all proper ancestors of $x$, and $C^d$, the connected component containing all proper descendants of $x$. Moreover, the smallest preorder number in $C^d$ is greater than the preorder number of $x$ in $T$. Thus, the subtree of $\T$ corresponding to $C^u$ must appear as the left subtree of the subtree of $\T$ corresponding to $C^d$.
	To prove statement $3$, suppose the node $z = \lca_{\T}( u, v)$ lies in $P_j$ with $i \neq j$.  
	Since $z \neq u$ and $z \neq v$, 
	the vertices $u$ and $v$ must belong to distinct connected components of $T \setminus \{z\}$, implying  $z \in P_j \cap P_T(u,v)$.  
	This contradicts the fact that $P_i$ and $P_j$ are disjoint, as $P_T( u, v)$ is a subgraph of  $P_i$ from the construction.
\end{proof}

Note that the paths $P_1, \dots, P_\ell$ can be identified by a single preorder traversal of $T$ (i.e., a lexicographical DFS from the vertex $r$): during the traversal, whenever we visit an adjacent vertex of a vertex $v$ other than the first one, we create a new path.
Next, for each $i \in [\ell]$, we define $\T_i$ as follows (see Figure~\ref{figure:general}(d) for an example):  
(1) for any leaf node $l$, pick the edge $e = (l,u)$ in $\T$ where $u$ is the parent of $l$;  
(2) if both $l$ and $u$ are in $P_i$, keep the edge; otherwise, contract the edge and assign the label 
of the merged node to either $l$ or $u$ (whichever belongs to $P_i$; if neither does, use $l$);  
(3) repeat steps (1) and (2) recursively until all nodes are in $P_i$.
Then, by Lemma~\ref{lem:tit}, $\T_i$ is a binary search tree on $P_i$, and given  the tree structures of $\T_1, \dots, \T_{\ell}$, their labels can be decoded by referring to $P_i$ using Lemma~\ref{lem:path}.  
As in the proof of Lemma~\ref{lem:path}, if a node in $\T_i$ has a single child, we treat that child as either the left or right child of its parent, depending on their preorder numbers in $T$. 

The representation of $\T$ is a single tree $\T'$, obtained from $\T$ by replacing each node’s label with the index of the path it belongs to. 
We store the tree structure of $\T'$ using BP, 
requiring $2n$ bits, and an array $I$ that stores the labels of $\T'$ in preorder.  
Using the representation of Dodis et al.~\cite{DBLP:conf/stoc/DodisPT10}, $I$ takes $\ceil{n \log \ell}$ bits, so the total space for (2) is $\ceil{n \log \ell} + 2n$ bits.  

To decode the label of a node $u \in \T$ from the representation,  
we first identify the path $P_i$ that contains $u$ using $\T'$.  
Then, we reconstruct $P_i$ from $T$ and $\T_i$ from $\T'$.  
The label of $u$ is obtained from $P_i$ and $\T_i$ using Lemma~\ref{lem:path}.

Finally, we analyze the construction time of the representation. We first construct a perfect hash table~\cite{fredman1984storing} that stores the mapping from the labels of the nodes in $\T$ to their position in $\T$. This can be done in $O(n)$ time by traversing $\T$ once. Next, as described earlier, we identify the paths $P_1, \dots, P_{\ell}$ in $O(n)$ time by a single preorder traversal of $T$. During this traversal, we can also construct $\T'$ by locating, for each visited vertex of $T$, its corresponding node in $\T$ in $O(1)$ time using the hash table.
Moreover, a BP of any tree of size $s$ can be built in $O(s)$ time, which implies that the bit string containing the BP of $\T'$ can be constructed in $O(n)$ time.
Finally, the array $I$ can be constructed in $O(n)$ time after constructing $\T'$ via the preorder traversal. We summarize the result in the following theorem.

\begin{theorem}\label{thm:general_upper}
	Given a tree $T$ with $n$ vertices and $\ell$ leaves, any search tree $\T$ on $T$ can be encoded using at most $\ceil{n \log \ell} + 2n$ bits. 
        The representation can be constructed in $O(n)$ time.
\end{theorem}

\begin{example}
	Consider how to decode the label of the node $u$ in Figure~\ref{figure:general}(c) using the representation of Theorem~\ref{thm:general_upper}. 
	Given that $u$ is labeled $2$ in $\T'$, we reconstruct $P_2$ from $T$, and the tree structure of $\T_2$ from $\T'$. Since $u$ is the first node in the inorder traversal of $\T_2$, the answer is the label of the first node in $P_2$ under the preorder traversal of $T$, which is $g$.
\end{example}

\section{Succinct Representation of Steiner-Closed Search Trees on Trees}\label{sec:steiner}
In this section we present a representation of a Steiner-closed search tree $\T$
on $T = (V, E)$ with $n$ vertices and $\ell$ leaves.
First, we treat $T$ as a rooted and ordered tree as follows:
(1) the root $r$ of $T$ is set to be the root of $\T$, and
(2) the siblings in $T$ are ordered according to their labels.
Therefore, by using $O(\log n)$ additional bits to indicate $r$, we can regard $T$ as a rooted and ordered tree.

The overall idea of the representation follows a strategy similar to the representation of Theorem~\ref{thm:general_upper},
where we decompose $T$ into $\ell$ distinct paths $P_1, \dots, P_{\ell}$ 
and construct the search trees $\T_1, \dots, \T_{\ell}$ on these paths by contracting edges of $\T$.
However, unlike Theorem~\ref{thm:general_upper}, where the paths are defined independently of $\T$,
here each path is determined based on the structure of $\T$.
In Lemma~\ref{lem:steiner-closed3},
we show that when $\T$ is Steiner-closed,
the subgraph of $\T$ induced by the nodes in $P_i$ forms a connected subgraph.
Using this property, we show that $\T$ can be reconstructed from $\T_1, \dots, \T_{\ell}$ without storing the path index of each node in $\T$,  
which is required in the representation of Theorem~\ref{thm:general_upper}. 
Before introducing the detailed decomposition procedure, we first present the following proposition, which shows that, when $\T$ is a Steiner-closed STT, the sibling order in $\T$ coincides with the order induced by the preorder numbers of the corresponding vertices in $T$.

\begin{proposition}\label{prop:siblingorder}
For any non-root node $u$ and its left sibling $v$ in a Steiner-closed STT $\T$ of $T$, $\preorder_T(v) < \preorder_T(u)$.
\end{proposition}
\begin{proof}
Suppose that $v$ is a left sibling of $u$ in $\T$, but $\preorder_T(v) > \preorder_T(u)$, and let $w$ be the parent of $u$ and $v$ in $\T$.
Let $v'$ be the node in $\T(v)$ with the smallest preorder number in $T$.
Since $v$ is the left sibling of $u$, the definition of the sibling ordering in STTs implies that $\preorder_T(v') < \preorder_T(u)$. Hence, $v' \neq v$, and $u$ cannot be an ancestor of $v$ in $T$ (otherwise, $v' \not\in \T(v)$).
Furthermore, the two vertices $z = \lca_{T}(u, v)$ and $z' = \lca_{T}(v, v')$ satisfy that $z'$ is an ancestor of $z$ in $T$. Otherwise, suppose that $z'$ is a proper descendant of $z$. Then there are two cases.
If $\preorder_T(u) < \preorder_T(z')$, then
$\preorder_T(u) < \preorder_T(z') < \preorder_T(v')$,
contradicting the assumption that $\preorder_T(v') < \preorder_T(u)$.
On the other hand, if $\preorder_T(u) > \preorder_T(z')$, then $\preorder_T(u)$ is greater than the preorder numbers of all vertices in $T(z')$, and in particular,
$\preorder_T(u) > \preorder_T(v)$,
contradicting the assumption that $\preorder_T(u) < \preorder_T(v)$.
Moreover, $w \neq z$ and $w \in P_T(z, u)$ (otherwise, from the properties of $z$ and $z'$, $v' \not\in \T(v)$), which implies that $z \in CH_T(\{r, v, w\})$.

Now let $H$ be the set of nodes on the path $P_{\T}(r, v)$.  
Then $v, w \in H$, but $z, z' \notin H$ (otherwise, $v' \not\in \T(v)$).  
Thus, $CH_T(H) \setminus H$ contains the vertex $z$, which has degree at least $3$ in the subtree of $T$ induced by $CH_T(H)$. 
This contradicts the assumption that $\T$ is Steiner-closed.
\end{proof}

\subparagraph*{Decomposition of $T$ into paths. }
As described earlier, we decompose $T$ into $\ell$ distinct paths $P_1, \dots, P_{\ell}$, 
which are defined based on the structure of $\T$.
The decomposition procedure consists of two main steps:
(i) \textit{Initial decomposition}, and
(ii) \textit{Modification procedure}.
Briefly, the initial decomposition partitions $T$ into paths by traversing $\T$ in preorder of $\T$, and the modification procedure refines these paths so that each path contains exactly one leaf of $T$ and all vertices in a path have distinct depths, as in the paths defined in Section~\ref{sec:general}.

The initial decomposition is performed as follows.
First, we initialize the vertex sets $V_1 = V_2 = \dots = V_{\ell} = \emptyset$, 
and traverse $\T$ in preorder. Whenever we traverse the node $u \in \T$, we find the smallest index $i \in [\ell]$ such that the induced subgraph of $T$ on $CH_T(V_i \cup \{u\})$ is a path and does not include any vertices belonging to sets other than $V_i$, and add $u$ to $V_i$. After the traversal, we define $P_i = T[V_i]$.
By construction, $V_1, \dots, V_{\ell}$ form a partition of $V$. 
The following three lemmas show that the initial decomposition is well-defined and that each $\T[V_i]$ forms a connected subgraph of $\T$ (the proofs are presented in \Cref{app:steiner-proofs}).

\begin{restatable}{lemma}{firstlemma}
\label{lem:steiner-closed1}
After the initial decomposition, for every $i \in [\ell]$: $P_i$ is a path in $T$.
\end{restatable}

\begin{restatable}{lemma}{secondlemma}
\label{lem:steiner-closed2}
After the initial decomposition, for every path $P_i$ in $T$, every vertex in $P_i$ that is not a leaf in $T$ has at least one child in $T$ that also lies in $P_i$. 
Consequently, $P_i$ contains at least one (and at most two) leaves of $T$ for all $i \in [\ell]$.
\end{restatable}

\begin{restatable}{lemma}{thirdlemma}
\label{lem:steiner-closed3}
$\T[V_i]$ is a connected subgraph of $\T$ for all $i \in [\ell]$.
\end{restatable}		
    After the initial decomposition, we apply the following modification procedure to the paths to ensure that
    (a) $P_i$ contains exactly one leaf node $l_i$ of $T$ (recall that $l_i$ is the $i$-th leaf in $T$, ordered by its preorder number in $T$), and
    (b) all vertices in $P_i$ have distinct depths in $T$, as in the decomposition of Section~\ref{sec:general}.

    We decompose each path $P_i$ that contains two leaves of $T$ into two paths as follows.
	Let $M_L$ and $M_R$ be the sets of vertices of $P_i$ in the left and right subtrees of $T(h_i)$, respectively,
	where $h_i$ is the LCA of the two leaves.
	Let $u$ be the first node in $V_i$ according to the preorder traversal of $\T$.
	If $u = h_i$, we arbitrarily assign $h_i$ to $V_L$. 
	Thus, if $u = h_i$ or $u \in M_L$, we define $V_L = M_L \cup \{h_i\}$ and $V_R = M_R$. 
	Otherwise, we define $V_L = M_L$ and $V_R = M_R \cup \{h_i\}$.

	Finally, we decompose $P_i$ into two paths: $T[V_L]$ and $T[V_R]$. 
    Once all paths contain only one leaf, we relabel them so that each $P_i$ denotes the path in $T$ containing $l_i$.

    We now claim that in each decomposition each $\T[V_L]$ and $\T[V_R]$ are connected subgraphs of $\T$, thus, maintaining the property of Lemma~\ref{lem:steiner-closed3}.
    Suppose not, and without loss of generality assume that $u \in V_L$.
    If $u = h_i$, the claim holds immediately. Otherwise, there exist a descendant $v \in V_R$ of $u$ in $\T$ and, 
    without loss of generality, let $v$ be the node with the minimum preorder number in $\T$ among all nodes in $V_R$.
    If $P_\T(r,v)$ contains $h_i$, then $h_i$ must be the last node in $V_L$ according to the preorder traversal of $\T$; otherwise, by the definition of $T$, $\T[V_L]$ would contain no nodes from $V_R$. 
	In this case, both $\T[V_L]$ and $\T[V_R]$ are connected subgraphs of $T$. 
	Otherwise, $P_\T(r,v)$ contains $u$ but not $h_i$, which contradicts the Steiner-closed property of $\mathcal T$ (recall that $r$ is also the root of $T$).
	Furthermore, by Lemma~\ref{lem:steiner-closed2}, both $T[V_L]$ and $T[V_R]$ contain exactly one leaf of $T$.  

\begin{example}
	Consider a tree $T$ and its Steiner-closed search tree $\T$ in  Figure~\ref{figure:steiner}(a) and (b), respectively. We first define the root of $T$ as $d$, which is a root of $\T$.
	The initial decomposition decomposes $T$ into five disjoint paths: $V_1 = \{d, b, f, g\}$, $V_2 = \{c, a\}$, $V_3 = \{i, e\}$, $V_4 = \{h\}$, and $V_5 = \emptyset$.
	Next, the modification procedure partitions $P_1$ into two paths: one containing $\{d, b\}$ and the other containing $\{f, g\}$, and relabels the paths $P_1 \dots, P_5$, as shown in Figure~\ref{figure:steiner}(a).  
\end{example}

\begin{figure}[t]
	\begin{center}
		\includegraphics[scale=0.65]{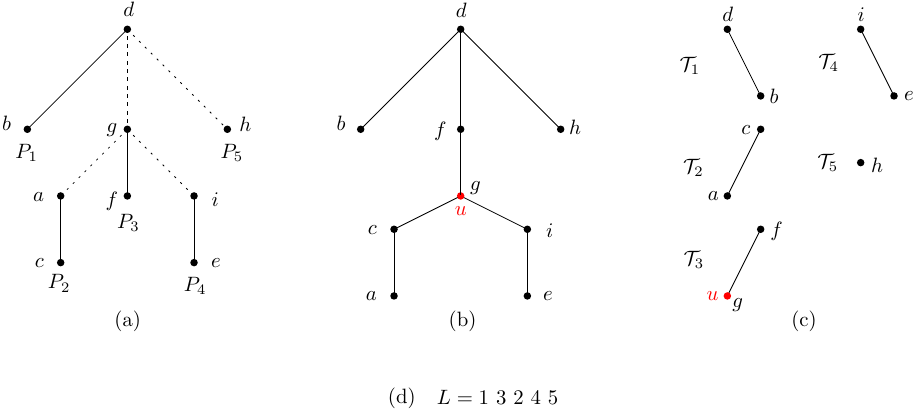}
	\end{center}
	\caption{(a) A $T$ with paths $\{P_1, \dots, P_5\}$ (dotted lines indicate edges that are not part of any path in $P$), 
		(b) Steiner-closed search tree $\T$ on $T$, 
		(c) the shapes of $\T_1, \dots, \T_5$, 
		(d) an array $L$. The representation of $\T$ given in Theorem~\ref{thm:steiner} consists of the shapes of $\T_1, \dots, \T_5$ (without the labels), the RMQ encoding of $L$, and the position of the root node in $T$ ($T$ is initially an unrooted tree).
		}
	\label{figure:steiner}
\end{figure}

\subparagraph*{Representation of $\T$. }
We define $\T_1, \dots, \T_{\ell}$ 
as the trees derived from the paths $P_1, \dots, P_{\ell}$, respectively, as in Section~\ref{sec:general}.  
By Lemmas~\ref{lem:steiner-closed1} and \ref{lem:steiner-closed3}, $\T_1, \dots, \T_{\ell}$ are search trees on their corresponding paths and form connected subgraphs of $\T$ (note that this latter property does not generally hold for arbitrary search trees on trees).

Now we describe the representation of $\T$.  
Let $s_i$ be the vertex in $P_i$ with the minimum distance from the root of $T$, i.e., $P_i = P_T( s_i, l_i)$. 
We first store the tree structures of $\T_1, \dots, \T_{\ell}$ using the representation of Lemma~\ref{lem:path}.  
By concatenating these representations into a single bit string in the order of the preorder numbers of $s_i$'s in $T$, this requires $2n - 2\ell$ bits in total.

Unlike the paths defined in Section~\ref{sec:general}, we cannot decode the paths from $T$ without additional information.  
However, using the properties of Lemmas~\ref{lem:steiner-closed1}, \ref{lem:steiner-closed2} and \ref{lem:steiner-closed3}, we can decode them using only a small extra space as follows.  
We first explicitly store the position of the root $r$ of $\T$ (which is also the root of $T$) using $\ceil{\log n}$ bits.  
Next, imagine an array $L[1, \dots, \ell]$ of size $\ell$,  
where $L[i] = j$ if and only if $s_i$ is the $j$-th entry in the sequence of vertices $S = (s_1, \dots, s_{\ell})$, ordered by the preorder traversal of $T$ (see Figure~\ref{figure:steiner}(d) for an example).  
Storing array $L$ would require too much space. Instead let us show that by performing {\em Range Minimum Queries (RMQ)} on $L$,\footnote{The RMQ is defined as follows: for $1 \le i \le j \le \ell$, $\RMQ{}(i,j)$ returns the position of the minimum value in the subarray $L[i, \dots, j]$} we can determine all paths $P_1, \dots, P_\ell$.

First, we reconstruct $P_{\min} = P(r, l_{\min})$ by computing $\min = \RMQ{}(1, \ell)$.  
Let $S_L$ (resp., $S_R$) be the set of all left (resp., right) siblings of the vertices in $P_{\min}$.  
Since all paths are disjoint, $S_L \subseteq \{s_1, \dots, s_{\min-1}\}$ and $S_R \subseteq \{s_{\min+1}, \dots, s_{\ell}\}$.  
Let $s_{left}$ be the vertex in $\{s_1, \dots, s_{\min-1}\}$ with the smallest preorder number in $T$, 
and define $s_{right}$ analogously. Then since $s_{left} \in S_L$ and $s_{right} \in S_R$, we can reconstruct $P_{left}$ by computing $left = \RMQ{}(1, \min-1)$ and $P_{right}$ by computing $right = \RMQ{}(\min+1, \ell)$. We repeat this procedure recursively until all paths $P_1, \dots, P_{\ell}$ are determined.

Thus, we do not need to store $L$ explicitly, but rather we just need a data structure that answers RMQ on $L$. Bender and Farach-Colton~\cite{bender2000lca} showed that RMQ queries can be answered via LCA queries on the Cartesian tree of $L$. Therefore, we obtain the RMQ data structure by storing the $2\ell$-bit BP representation of the shape of the Cartesian tree on $L$~\cite{DBLP:journals/cacm/Vuillemin80}.

After reconstructing the paths, we reconstruct the search trees $\T_1, \dots, \T_{\ell}$ with their labels by reading the bit string of their representations sequentially.  
Finally, we merge the search trees into a single tree $\T$.  
By Lemma~\ref{lem:steiner-closed3} and the definition of the sibling ordering in $\T$, we can merge them according to the preorder numbers of the nodes in $T$.

From the above construction it follows that the total space required for the representation of a Steiner-closed STT is $(2n - 2\ell) + \ceil{\log n} + 2\ell = 2n + \ceil{\log n}$ bits.

This representation can be constructed in $O(n\ell)$ time as follows:  
(1) construct a mapping between the vertices in $T$ and the nodes in $\T$ in $O(n)$ time by traversal; (2) perform the initial decomposition and modification procedure on $T$ in $O(n\ell)$ time, since for each node visited during the decomposition we can determine its vertex set in $O(\ell)$ time by checking each set in $O(1)$ time using LCA on $T$; (3) construct $L$ (which can be built during the modification procedure) and $\T_1, \dots, \T_{\ell}$ in $O(n)$ time; and (4) construct the representations of the structures in (3) in $O(n)$ time~\cite{DBLP:journals/siamcomp/FischerH11}.

We summarize the result in the following theorem:

\begin{theorem}\label{thm:steiner}
	Given a labeled and ordered tree $T$ with $n$ vertices and $\ell$ leaves, any Steiner-closed search tree $\T$ on $T$ can be encoded using $2n + \ceil{\log {n}}$ bits. The representation can be constructed in $O(n\ell)$ time. 
\end{theorem}

\begin{example}
	Consider how to decode the label of the node $u$ in Figure~\ref{figure:steiner}(b) by reconstructing $\T$ from the representation.  
	First, we decode the path $P_{\min} = P_1$.  
	Here, $S_L = \emptyset$, $S_R = \{g, h\}$, and $right = \RMQ{}(2,5) = 3$, which implies $P_3 = P(g,f)$.  
	By repeating this procedure, we construct all paths $P_1, \dots, P_5$.  
	Next, we decode the tree structures with the labels of $\T_1, \dots, \T_5$ and merge them into a single tree $\T$. For example, to merge $\T_1$ and $\T_3$, we first check the preorder numbers of $b$ and $f$ in $T$, and since $\preorder_T(b) < \preorder_T(f)$, we attach the root of $\T_3$ (the node $f$) as the child of $d$ and as the right sibling of $b$, following the definition of the sibling ordering in $\T$.  
	After fully reconstructing $\T$, we can determine that the label of $u$ is $g$.
\end{example}

\section{Conclusions}

In this paper, we present space-efficient representations of general and Steiner-closed STTs
that can be constructed in polynomial time and achieve optimal space up to lower-order
additive terms, as well as space-efficient data structures supporting fast traversals.

Natural directions include extending our techniques to $k$-cut trees~\cite{splayTT} and designing space-efficient data structures for SplayTT~\cite{splayTT} that support rotations using tiny pointers~\cite{DBLP:conf/soda/BenderCFKT23}. Another interesting direction is to obtain instance-optimal succinct
representations for general STTs whose space depends on the topology of $T$, since the
$n\log\ell$ bound need not be tight for every $T$.
\newpage
\bibliography{ref}
\newpage
\appendix


        

\section{Proofs of the Lemmas in \Cref{sec:steiner}.}\label{app:steiner-proofs} 
Here, we provide the proofs of Lemmas~\ref{lem:steiner-closed1}, \ref{lem:steiner-closed2}, and \ref{lem:steiner-closed3}.

   \begin{figure}[ht]
	\begin{center}
	\includegraphics[scale=0.9]{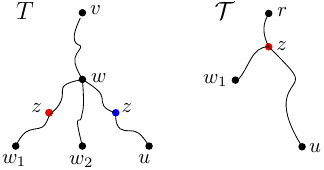}
	\end{center}
	\caption{Supplementary figures for the proof of Lemma~\ref{lem:steiner-closed1}.}
	\label{figure:lemma5}
    \end{figure}

\firstlemma*
\begin{proof}
        By construction, every $CH_T(V_i)$ is a path. Thus, to prove the first property, we need to show that for every pair $u, v \in V_i$ there does not exist another vertex $w \in P_T( u, v) \cap V_j$ with $j \neq i$. 
        Assume for contradiction that such $w$ exist for some pair $u,v \in V_i$ 
        and without loss of generality let $u$ be traversed before $v$ in preorder traversal of $\T$. 
		By the construction of $P_i$ and $P_j$, $w$ must be traversed after both $u$ and $v$ during the preorder traversal of $\T$ (otherwise, either (a) $u$ and $v$ would not be placed in the same path $P_i$, or (b) $u$, $v$, and $w$ would belong to the same path). This implies that $j < i$ (otherwise, $w$ would have been placed in $P_i$). Now we prove the following claim.

        \begin{claim}
            $w$ has at least $4$ neighbors in $T$.
        \end{claim}
        \begin{proof}
        $P_j$ must have consisted of at least two vertices $w_1, w_2 \in V_j$ at the time when $u$ was visited, and $CH_T(\{w_1, w_2, w\})$ is a path in $T$, whereas $CH_T(\{w_1, w_2, u\})$ is not a path in $T$ 
        (otherwise, $u$ would have been added to $V_j$, since it is traversed before $w$). Also, if $w \not \in P_T( w_1, w_2)$, then $CH_T(\{w_1, w_2, w, u\})$ would also have to be a path, because $w \in P_T(u,v)$ and $CH_T(\{w_1, w_2, w\})$ is a path in $T$. However, this would imply that $CH_T(\{w_1, w_2, u\})$ is also a path, which we said is impossible. Therefore, $w \in P_T( w_1, w_2)$ and recall that by our assumption $w$ belongs to $P_T( u, v)$, yet $CH_T(\{w_1, w_2, u\})$ is not a path in $T$. Thus, $w$ has at least $4$ neighbors ($u$, $v$, $w_1$, and $w_2$) in $T$. 
        \end{proof}
        
        Since the claim holds, $w$ must be $\lca_T( x, y)$ for some $x \in \{u, v\}$ and $y \in \{w_1, w_2\}$.
        Without loss of generality, assume that $w = \lca_T( w_1, u)$, and let $z = \lca_{\T}( w_1, u)$.
        Then $z \neq w$, because both $u$ and $w_1$ precede $w$ in the preorder traversal in $\T$ and $z \in P_\T( w_1, u)$. 
        Moreover, by the definition of $\T$, the vertex $z$ lies either on the path $P_T( w, w_1)$ or on the path $P_T( w, u)$ in $T$ (see Figure~\ref{figure:lemma5}). 

        In case $z$ lies on the path $P_T( w, w_1)$, consider the set of nodes $V'$ that form the path $P_\T( r, u)$. Observe that $V'$ contains $z$ (and $u$), but cannot contain $w$, because $\preorder_{\T}(u) < \preorder_{\T}(w)$ and $u$ is a descendant of $z$ in $\T$.
        This implies that $w \in CH_T(V') \setminus V'$, and, as was shown before, $w$ has degree greater than 2 in $T$, which contradicts the Steiner-closed property of $\T$.

        In case if $z$ lies on the path $P_T( w, u)$, we can apply the same argument to show that the set of nodes $V''$ that form the path $P_\T( r, w_1)$  contains $z$ (and $w_1$), but not $w \in CH_T(V'') \setminus V''$, again contradicting the Steiner-closed property of $\T$.
\end{proof}

 \begin{figure}[ht]
	\begin{center}
	\includegraphics[scale=0.9]{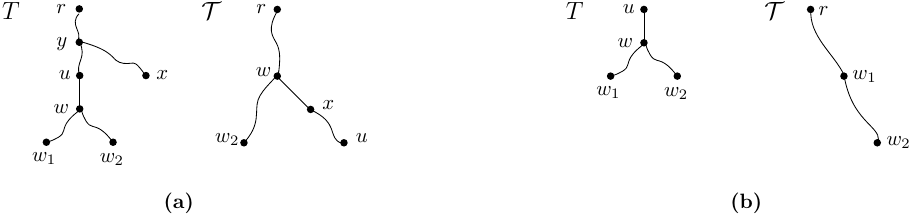}
	\end{center}
	\caption{Supplementary figures for the proof of Lemma~\ref{lem:steiner-closed2}.}
	\label{figure:lemma6}
    \end{figure}

\secondlemma*
\begin{proof}
    Assume for contradiction that $u \in V_i$ is a non-leaf vertex of $T$ with a single child $w \in V_j$ and $j \neq i$ (the case where $u$ has multiple children can be handled analogously under the assumption that none of them belongs to $V_i$).
    If $i < j$, $w$ would have been placed in $V_i$ because the set $CH_T(V_i \cup \{w\})$ forms a single path in $T$. Therefore, it must be that $j < i$.
    Furthermore,the vertices $V_j$ form a (connected) path in $T$ by Lemma~\ref{lem:steiner-closed1}. This implies that all vertices in $V_j$ must lie in $T(w)$ because $w$'s parent $u \notin V_j$.
        
    Thus, there exist two vertices $w_1$ and $w_2$ in $V_j \setminus \{w\}$ with $\lca_T( w_1, w_2) = w$ (otherwise, since $CH_T(\{w, w_1, w_2\})$ is a path in $T$, $CH_T(\{u, w, w_1, w_2\})$ would also be a path in $T$ and $u$ would have been added to $V_j$ during the construction). Therefore, $w$ has at least 3 neighbors in $T$. Now, without loss of generality, let $w_1$ and $w_2$ be the {\em first} two nodes in $V_j \setminus \{w\}$ in the
    preorder traversal of $\T$ that satisfy $\lca_T( w_1, w_2) = w$, 
    and let $\preorder_{\T}(w_1) < \preorder_{\T}(w_2)$.
    Then the following claim holds.

    \begin{claim}
        $\preorder_{\T}(w_2) < \preorder_{\T}(w)$.
    \end{claim}
    \begin{proof}[Proof (of claim)]
         Assume for contradiction that, $\preorder_{\T}(w_2) > \preorder_{\T}(w)$. In this case, it is enough to show that $\preorder_{\T}(u) < \preorder_{\T}(w_2)$. If this holds, then since $CH_T(\{u, w, w_1\})$ forms a single path in this case, $u$ would have been placed in $V_j$ instead of $V_i$, which contradicts to the fact that $u \in V_i$. Therefore, it must hold that $\preorder_{\T}(w_1) < \preorder_{\T}(w_2) < \preorder_{\T}(w)$, which proves the claim.

        To prove $\preorder_{\T}(u) < \preorder_{\T}(w_2)$, 
        assume for contradiction that $\preorder_{\T}(u) > \preorder_{\T}(w_2) > \preorder_{\T}(w)$. 
        Then $u$ is a descendant of $w$ in $\T$, since $u$ and $w$ lie in the same connected component of $T \setminus \{a\}$ for all $a \in V$ other than $u$ and $w$. Similarly, $w_2$ is also a descendant of $w$ in $\T$, since $w$ and $w_2$ lie in the same connected component of $T \setminus \{a\}$ for all $a \in V$ other than the vertices in $P_T( w, w_2)$. In this case, such a vertex $a$ belongs to $V_j$ by Lemma~\ref{lem:steiner-closed1}, and we choose $a = w_2$.
        Thus, $u$ and $w_2$ lie in different child subtrees of $\T(w)$, 
        and this follows that $u$ lies in a child subtree of $\T(w)$ that is to the right of the child subtree containing $w_2$. Let $x$ be the child of $w$ in $\T$ such that the subtree $\T(x)$ contains $u$ and 
        let $y = \lca_T( w, x)$.
        Then 
        the following properties hold (see Figure~\ref{figure:lemma6} (a)): 

        \begin{itemize}
            \item $y$ is neither $w$ nor $x$. This is because the vertices $u$ and $x$ cannot be exist in the same connected component of in $T \setminus \{w\}$ if $x \in T(w)$. Also if $y = x$, $x$ is a left sibling of $w_2$ because siblings in $\T$ are ordered by their preorder numbers in $T$. 

            \item $y \neq r$ (otherwise, $x$ cannot be a child of $w$ in $\T$ from the above property). This implies $y$ has a degree greater than $2$. 
            
            \item The path $P_\T( r, x)$ does not contain the node $y$. Otherwise, $x \not\in \T(w)$ since $w$ and $x$ lie in different connected components of $T \setminus \{y\}$ from the above properties.
        \end{itemize}

        Recall that $x$ is a child of $w$ in $\T$, i.e., the path $P_\T( r, x)$ contains both $w$ and $x$, 
        (b) establishes that it does not contain $y$ that has degree greater than $2$, i.e., it contradicts that $\T$ Steiner-closed. Hence, $\preorder_{\T}(u) < \preorder_{\T}(w_2)$.
    \end{proof}

    Now consider the set of nodes $V'$ that form the path $P_\T( r, w_2)$. Clearly, $w_2 \in V'$, and moreover, the following claim holds.
    
    \begin{claim}
        $w_1$ is an ancestor of $w_2$ in $\T$, i.e., $w_1 \in V'$.
    \end{claim}
    \begin{proof}[Proof (of claim)]
        Assume for contradiction that this is not the case, i.e., some $z = \lca_{\T}( w_1, w_2) \neq w_1$. Then $\preorder_{\T}(z) < \preorder_{\T}(w_1) < \preorder_{\T}(w_2)$. 
        There are two cases to consider: either $z = w$ or $z \in V_j \setminus \{w\}$. Clearly, it cannot be the first case because we have established that $\preorder_{\T}(w) > \preorder_{\T}(w_2) > \preorder_{\T}(w_1)$. Furthermore, the second case contradicts our choice of $w_1$ and $w_2$ as being the {\em first} two nodes in $V_j \setminus \{w\}$ in the preorder traversal of $\T$ that satisfy $\lca_T( w_1, w_2) = w$. Therefore, $w_1 = \lca_{\T}( w_1, w_2) \in V'$.
    \end{proof}

    Finally, $w$ cannot be in $V'$ because $\preorder_{\T}(w_2) < \preorder_{\T}(w)$, and since $w = \lca_T( w_1, w_2)$, i.e., $w \in P_T( w_1, w_2)$ (see Figure~\ref{figure:lemma6} (b) for an illustration). It follows that $w \in CH_T(V') \setminus V'$ and as was shown earlier, $w$ has degree greater than 2 in $T$, which contradicts the Steiner-closed property of $\T$.
\end{proof}

\begin{figure}[ht]
	\begin{center}
	\includegraphics[scale=0.9]{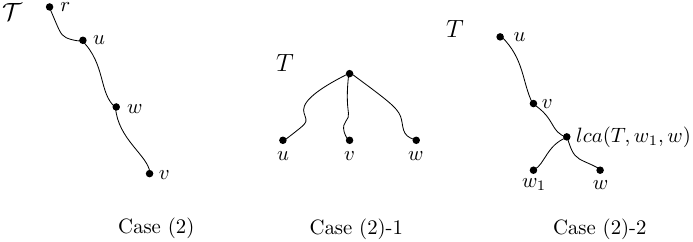}
	\end{center}
	\caption{Supplementary figures for the proof of Lemma~\ref{lem:steiner-closed3}.}
	\label{figure:lemma7}
    \end{figure}

\thirdlemma*
\begin{proof}
    Assume for contradiction that there exist two nodes $u, v \in V_i$ and a node $w \in P_\T( u, v) \cap V_j$ with $i \neq j$.
    Without loss of generality, assume that $u$, $v$, and $w$ are the first three nodes in preorder traversal of $\T$ that satisfy this property, and let $\preorder_{\T}(u) < \preorder_{\T}(v)$. Then $w \notin P_T( u, v)$ by Lemma~\ref{lem:steiner-closed1}.
    We consider the following two cases:
    Case (1) $u \neq \lca_{\T}( u, v)$, and
    Case (2) $u = \lca_{\T}( u, v)$, i.e., $w$ is a descendant of $u$ and an ancestor of $v$ in $\T$.

    In Case (1), let $w' = \lca_{\T}( u,v)$. If $w' \neq w$, then it violates our choice of $u$, $v$, and $w$ as being the {\em first} three nodes in preorder traversal of $\T$ that satisfy the property that $w \in P_\T( u, v) \cap V_j$ with $i \neq j$. Thus, $w = \lca_{\T}( u,v)$. But then, the vertices $u$ and $v$ would lie in different connected components of $T \setminus \{w\}$, and it follows that $w \in P_T( u, v)$, which is a contradiction.

    In Case (2), 
    we further distinguish two following subcases as:
    Case (2)-1, $CH(\{u, v, w\})$ is not a path in $T$, and
    Case (2)-2,  $CH(\{u, v, w\})$ is a path in $T$ (see Figure~\ref{figure:lemma7}).
    In Case (2)-1, without loss of generality, assume that $\preorder_T(u) < \preorder_T(v)$ (otherwise, we simply swap $u$ and $v$ in the remainder of the proof). Then the node $\lca_T( w, v) \neq w$ is contained in neither $P_\T( r, w)$ 
    (otherwise, $w$ would not lie on the path $P_\T( r, v)$), nor $P_\T( w, v)$ 
    (otherwise, we choose $v$ as $\lca_T( w, v)$, which falls into Case (2)-2).
    Thus, the set of nodes  $V'$ forming $P_\T( r, v)$ contains $v$ and $w$ but does not contain 
    $\lca_T( v, w)$, which has at least 3 adjacent vertices in $T$, and $\lca_T( v, w) \in CH_T(V') \setminus V'$. This contradicts to  the Steiner-closed property of $\T$.

    Finally, consider Case (2)-2. Since $P_T( u, v)$ does not contain $w$ by Lemma~\ref{lem:steiner-closed1}, 
    we can distinguish three configurations of $T$:
    (i) $v = \lca_T( u, w)$;
    (ii) $u$ is an ancestor of $v$ and $v$ is an ancestor of $w$; or
    (iii) $w$ is an ancestor of $v$ and $v$ is an ancestor of $u$, according to the configuration of $\T$ in Case (2).
    In Case (i), the set of nodes $V'$ forming $P_\T( r, w)$ contains $u$ and $w$ but does not contain $v$, where $v$ has at least three adjacent vertices in $T$ (note that $v$ cannot be the root of $T$, which is also the root of $\T$). This contradicts the Steiner-closed property of $\T$.

    In Case (ii) and (iii), since $w$ is traversed before $v$ in the preorder traversal of $\T$, there exists at least one node $w_1 \in V_j$ satisfying the following conditions: (a) $w_1$ is traversed before $w$ in the preorder traversal of $\T$ (otherwise, both $u$ and $w$ would belong to either $V_i$ or $V_j$), and (b) the vertex $\lca_T( w_1, w)$, which is neither $w$ nor $w_1$, lies on the path $P_T( v, w)$ by property (1) (if no such $w_1$ exists, $CH_T(\{u, w, w_1\})$ would form a single path in $T$, implying that both $u$ and $w$ belong to either $V_i$ or $V_j$).  
    Without loss of generality, let $w_1$ be the first node in the preorder traversal of $\T$ that satisfies both (a) and (b).         
    Then $P_\T( r, w)$ does not contain any node on the path $P_T( v, w)$ that includes $\lca_T( w_1, w)$ (otherwise, $v$ would not be a descendant of $w$ in $\T$).
    Thus, the node $w_1$ can exist only in either $P_\T( r, u)$ (in which case we choose $w_1$, $u$, and $w$ instead of $u$, $x$, and $v$, respectively) or $P_\T( u, w)$ (in which case we choose $w_1$ instead of $w_2$). Both cases contradict the way in which we chose the nodes $u$, $w$, and $v$. 
\end{proof}

\section{Space Lower Bounds}\label{sec:lb}
In this section we provide space lower bounds for storing general STTs and Steiner-closed STTs.

\subparagraph*{Space lower bound for general STTs. } The following theorem shows that the Theorem~\ref{thm:general_upper} gives an optimal representation of $\T$ up to $O(n)$ additive terms.

\begin{theorem}\label{thm:general_lower}
For any $n$ and $2 \le \ell \le n/2$, there exists a rooted tree $T$ with $n$ vertices and $\ell$ leaves that requires at least $n \log {\ell}-O(\ell \log {(n/\ell)})$ bits to encode a search tree on $T$. 
\end{theorem}
\begin{proof}
      Let $t = \floor{n / \ell}$. We then construct a `comb-shaped' tree $T$ using $n$ vertices through the following steps. For simplicity, we first assume that $n$ is a multiple of $t$ and consider the general case later.
	First, construct $\ell$ distinct paths $Q_1, \dots, Q_{\ell}$, where for each $i \in [\ell]$,  
	the set of vertices in $Q_i$ is $\{(i-1)t+1, \dots, it\}\}$ and the set of edges in $Q_i$ is  
	$\{(j, j+1) \mid (i-1)t+1 \le j < it\}$.  
	Next, add $(\ell-1)$ additional edges $\{((i-1)t+1, it+1)\}$ for all $1 \le i < \ell$ to connect the $\ell$ paths  
	(see Figure~\ref{figure:general_lower} for an example).  
	Finally, relabel all the vertices in $Q_i$ as $i$.  
    From the construction, we can treat $T$ as a rooted tree with $\ell$ leaves by choosing any vertex with degree at least two.
	
	Now let $N_T$ be the set of all distinct search trees $\T$ on $T$ constructed as follows.
	Initially, the root $r$ of $\T$ is chosen as one of the vertices in $T$ with degree $1$.
	Then, recursively, the root of $\T \setminus \{r\}$ is chosen as one of the vertices in $T \setminus \{r\}$ with degree $1$.  
	Then $N_T$ is the set of all paths of size $n$, where exactly $t$ nodes among them are labeled as $i$ for each $i \in [\ell]$.  
    Therefore, the size of $N_T$ is $\frac{n!}{((n/\ell)!)^{\ell}}$, which implies that at least
    $\log |N_T| \ge n \log \ell - O(\ell \log {(n/\ell)})$ bits are necessary to encode a search tree in $N_T$ by Stirling’s approximation\footnote{$\log {\frac{n!}{((n/\ell)!)^{\ell}}} \approxeq n \log n - cn +O(\log n) - \ell((n/\ell)\log (n /\ell) - c(n/\ell) + O(\log (n/\ell))$ for some constant $c = \log e$}. Note that this is at least $n \log \ell - O(n)$ when $\ell = \Theta(n)$.

    Finally, consider the case when $n$ is not a multiple of $t$.
    In this case, let $n' = \floor{n/l} \cdot l$. i.e., the largest value at most $n$ that is a multiple of $t$ or $\ell$.
    Then $n' \log \ell - O(\ell \log (n'/\ell)) \ge (n - \min(\ell, n/\ell)) \log \ell - O(\ell \log (n/\ell))$. When $\ell \le n/\ell$, we have $\ell \log \ell \le \ell \log (n/\ell)$; otherwise, the term $\ell \log (n/\ell)$ dominates the term $\frac{n}{\ell} \log (n/\ell)$. Thus, the bound is still at least $n \log \ell - O(\ell \log (n/\ell))$.
\end{proof}

\subparagraph*{Space lower bound for Steiner-closed STTs. } The following theorem that shows that the representation of Theorem~\ref{thm:steiner} gives an optimal representation of $\T$ up to lower-order additive terms.

\begin{theorem}\label{thm:steiner_lower}
	For any tree $T$ with $n$ nodes, at least  $2n - O(\log {n})$ bits are necessary to encode Steiner-closed search tree on $T$.  
\end{theorem}
\begin{proof}
	Consider a path $T$ of size $n$.
	Every search tree on $T$ is Steiner-closed since the maximum degree of the vertices in $T$ is $2$.
	Moreover, from the argument used in the proof of Lemma~\ref{lem:path}, search trees on $T$ are in one-to-one correspondence with standard BSTs on an array of size $n$. Hence, there are exactly $C_n$ Steiner-closed search trees on $T$, where $C_n$ is the $n$-th Catalan number.
	Therefore, at least $\log C_n = 2n - O(\log n)$ bits are necessary to encode a Steiner-closed tree.
\end{proof}

\section{Data Structure for $\access$}\label{sec:access}
In this section we present asymptotically optimal-space data structures for a search tree $\T$ on $T$ to support $\access(k)$ efficiently, which returns the label of the node in $\T$ with preorder number $k$. We denote this node by $u_k$. 

As in the previous sections, we consider two cases:
when $\T$ is a general search tree and when $\T$ is a Steiner-closed search tree.
For both cases, we assume that $T$ is a rooted and ordered tree, defined in the same way as in the previous sections, and that basic navigational operations on $T$ can be performed in $O(1)$ time using $2n + o(n)$ bits of extra space~\cite{DBLP:journals/talg/NavarroS14}
(see Table~1 in \cite{DBLP:journals/talg/NavarroS14} for the list of supported operations).
We also define the paths $P_1, \dots, P_{\ell}$ on $T$ and their corresponding search trees $\T_1, \dots, \T_{\ell}$ as in the previous sections.

Before presenting the data structures, we briefly describe the additional auxiliary data structures used in this paper.

\subparagraph{Rank and Select. }
Consider an array $S[1,n]$ over an alphabet $\Sigma$ of size $\sigma$.  
Given a character $c \in \Sigma$ and an index $i \in [n]$,  
the \textit{rank} operation returns the number of occurrences of $c$ in the prefix $S[1,i]$,  
and the \textit{select} operation returns the position of the $i$-th occurrence of $c$ in $S$.  
Golynski, Munro, and Rao~\cite{DBLP:conf/soda/GolynskiMR06} showed that there exists 
a space-efficient data structure for $S$ that supports these operations efficiently:

\begin{lemma}[\cite{DBLP:conf/soda/GolynskiMR06}]\label{lem:rankselect}
	There exists an $(n \log \sigma + o(n \log \sigma))$-bit data structure that supports 
	rank and select operations on $S$ in $O(\log \log \sigma)$ and $O(1)$ time, respectively.  
	The data structure also supports accessing any position in $S$ in $O(\log \log \sigma)$ time.
\end{lemma}

\subparagraph{Balanced parenthesis representation. } Given a balanced parentheses sequence (BP), the \textit{findopen} operation returns the position of the matching open parenthesis when the position of a closed parenthesis is given. 
Analogously, one can define a \textit{findclose} operation. 
Munro and Raman~\cite{munro2001succinct} showed that there exists a data structure that supports these operations efficiently:

\begin{lemma}[\cite{munro2001succinct}]\label{lem:bp}
	Given a BP of size $n$, there exists an $o(n)$-bit auxiliary structure that can answer both findopen and findclose in $O(1)$ time.
\end{lemma}

Given a rooted tree $T$, a BP of $T$ is a representation of $T$ using a BP.  
Specifically, starting from an empty sequence, we perform an Euler tour on $T$ (if $T$ is an ordered tree, we visit the subtrees from left to right).  
For any vertex $v \in T$, append `(' when starting the tour of $T(v)$, and append `)' when finishing the tour of $T(v)$. Clearly, if $T$ has $n$ vertices, the size of the BP of $T$ is always $2n$ bits and the BP of $T$ can be constructed in $O(n)$ time.
Given a BP of $T$, one can support a wide range of tree navigational queries using auxiliary structures of $o(n)$ bits in total~\cite{DBLP:journals/talg/NavarroS14}.

\subsection{Data Structure for General Search Trees on Trees}
To support the $\access$ operation on $\T$, we build upon the representation of Theorem~\ref{thm:general_upper}  
(the labeled tree $\T'$ together with the representations of $\T_1, \dots, \T_{\ell}$ using Lemma~\ref{lem:path}),  
with the following modifications:

\begin{enumerate}[label={(\arabic*)}]
	\item In addition to the ($\ceil{n \log \ell}+2n$)-bit representation of $\T'$ (an array $I$ together with the BP representation of $\T'$),  
	we maintain an $o(n)$-bit auxiliary structure to support tree navigational queries on $\T'$ in $O(1)$ time~\cite{DBLP:journals/talg/NavarroS14}.  
	Moreover, we maintain $I$ using the data structure of Lemma~\ref{lem:rankselect},  
	requiring $o(n \log \ell)$ extra bits to support rank and select operations on $I$.
	
	\item We explicitly store the tree structures $\T_1, \dots, \T_{\ell}$ using Lemma~\ref{lem:path} and merge their representations into a single bit string of total size $2n$ bits (here, we do not delete the first and last parentheses as in Lemma~\ref{lem:path}).

	In addition, to enable random access to the bit string, we modify the bit string by inserting a special marker before the first bit of each new tree representation, and maintain the resulting sequence using the data structure of Lemma~\ref{lem:rankselect}.
	Since the modified string forms a ternary string of size at most $2n+\ell$, the total required space is $(2n+\ell)\ceil{\log 3} = O(n)$ bits.	
     
	\item For each search tree $\T_i$,  
	we maintain an $o(|\T_i|)$-bit auxiliary structure to support \textit{inorder rank} queries,  
	i.e., given the preorder number of a node, return its inorder number in $O(1)$ time~\cite{DBLP:journals/mics/DavoodiRS17}.  
	Overall, these auxiliary structures require $o(n)$ bits in total. 
\end{enumerate}

To answer $\access(k)$, we proceed as follows.  
First, we retrieve $I[k] = j$ in $O(\log \log \ell)$ time using the data structure of Lemma~\ref{lem:rankselect}.  
From the construction of $\T_j$, we compute $\preorder_{\T_j}(u_k)$ in $O(\log \log \ell)$ time using rank on $I$.  
Then, we compute $i_k=\inorder_{T_j}(u_k)$ using the inorder-rank structure.
Since $P_j$ ends at the $j$-th leaf $l_j$ of $T$, the vertex corresponding to
$u_k$ is the ancestor of $l_j$ at depth
$\depth_T(l_j)-|P_j|+i_k$.
Since $|P_j|$ can also be obtained in $O(\log\log\ell)$ time using rank on $I$,
$\access(k)$ can be answered in $O(\log\log\ell)$ time using
\texttt{leaf\_select} and level-ancestor queries on $T$~\cite{DBLP:journals/talg/NavarroS14}.
We summarize the result in the following theorem.

\begin{theorem}\label{thm:ds_general}
	There exists a data structure of size $n \log \ell + O(n) + o(n \log \ell)$ bits  
	that supports $\access$ on a search tree $\T$ on $T$ in $O(\log \log \ell)$ time,  
	where $n$ and $\ell$ are the number of vertices and leaves in $T$, respectively.
\end{theorem}

\begin{remark}
	The data structure of Theorem~\ref{thm:ds_general} supports $\access$ in $O(1)$ time when $\ell = O(\mathrm{polylog}(n))$ by using multiary wavelet trees, which can access any string of size $n$ over an alphabet of size $\ell$ in $O(\log \ell / \log \log n)$ time while supporting rank and select operations within the same time.    
\end{remark}

\subsection{Data Structure for Steiner-closed Search Trees on Trees}
In this section we present an $O(n)$-bit data structure for a Steiner-closed search tree $\T$ on $T$ that supports $\access$ in $O(1)$ time.

We first briefly describe the overall idea of the data structure.  
Unlike the general case, the representation of Theorem~\ref{thm:steiner} does not store the tree structure of $\T$,  but only the tree structures of $\T_1, \dots, \T_\ell$.  
Hence, the main challenge is to identify the path containing the query node  
without fully reconstructing $\T$, while using at most $O(n)$ extra bits.
To address this, we introduce an \textit{extended search tree} of size at most $2n$,  
which includes up to $n$ additional \textit{copied nodes} from $\T$.  
Using the extended search tree together with $O(n)$-bit auxiliary structures,  
we can determine the path that contains $u_k$ in $O(1)$ time by exploiting the property of Lemma~\ref{lem:steiner_ds1}. 
Furthermore, for any node in the extended search tree,  
we can find its corresponding node in $\T$ in $O(1)$ time using $O(n)$-bit auxiliary structures,
which allows us to support $\access$ in $O(1)$ time.

Let $P_i = P_T( s_i, l_i)$ be  
the path that contains the vertex $u_k$.  
Then we can answer $\access(k)$ by:  
(1) finding the vertices $s_i$ and $l_i$ (i.e., find their preorder numbers in $T$),  
(2) compute $\inorder_{\T_i}(u_k)$, and  
(3) returning $u_k$ (i.e., the label of the node) by returning the label of the ancestor of $l_i$ at depth  $\depth_T(s_i) + \inorder_{\T_i}(u_k) - 1$ in $T$.  
In this section we first define an extended search tree of $\T$ and describe the data structures required for steps (1) and (2). Note that step (3) can be performed in $O(1)$ time  
using tree navigational queries on $T$.

\subparagraph*{Extended search tree $\T^e$ of $\T$.}
We first introduce a lemma that will be useful throughout this section. 

\begin{lemma}\label{lem:steiner_left}
	For any tree $T$ and its Steiner-closed search tree $\T$ on $T$,  
	suppose $p_L$ is a child of $p$ in $\T$.  
	If $\preorder_T(p_L) < \preorder_T(p)$, then $p_L$ is an ancestor of $p$ in $T$,  
	which implies that $p_L$ is always the leftmost child of $p$ in $\T$.
\end{lemma}
\begin{proof}
	Assume $\preorder_T(p_L) < \preorder_T(p)$,  
	but $p_L$ is not an ancestor of $p$ in $T$.  
	Let $v = \lca_T(p_L, p)$.  
	Then $v \neq p_L$, and the path $P_\T( r, p_L)$ contains both $p$ and $p_L$ but not $v$; 
	otherwise, $p$ would not be included in $P_\T( r, p_L)$.  
	Thus, the set of nodes in $P_\T( r, p_L)$ would not form a Steiner-closed set,  
	leading to a contradiction.
\end{proof}

Now we define an \textit{extended search tree} $\T^e$ of $\T$ as follows.  
For each node $p \in \T$, let $p_L$ be its leftmost child,  
and let $C(p_L)$ denote the set of children of $p_L$.  
Then $\T^e$ is constructed from $\T$ by performing the following procedure for all nodes in $\T$:
\begin{enumerate}
	\item If $\preorder_T(p_L) < \preorder_T(p)$,  
	add a \textit{copied node} of $p_L$ as the rightmost child of $p$.  
	By Lemma~\ref{lem:steiner_left}, this occurs if and only if $p_L$ is an ancestor of $p$ in $T$.
	
	\item After creating the copied node of $p_L$,  
	for each node $c \in C(p_L)$, move the subtree $\T(c)$ to be a child subtree of the copied node of $p_L$ if $\preorder_T(p) < \preorder_T(c)$. If multiple subtrees are moved, their sibling order is preserved exactly as in $\T$.
\end{enumerate}

See Figure~\ref{figure:dummy} for an example.  
In the rest of this section, for each node $p \in \T^e$ that has a corresponding copied node,  
we say that $p$ is an \textit{original} node if it is not itself a copied node.

From the construction, $\T^e$ has at most $2n$ nodes,  
since each node in $\T$ appears at most twice in $\T^e$.  
Moreover, for any node $p$, the set of nodes on the paths $P_\T( r, p)$  
and $P_{\T^e}( r, p)$ are identical, except that some nodes in $P_\T( r, p)$ may be replaced by their copied versions in $\T^e$.  
Thus, the set of nodes in $P_{\T^e}( r, p)$ is also a Steiner-closed vertex set in $T$.  

We now introduce the following lemma,  
which enables us to compute $s_i$ and $l_i$ from $T$ and $\T^e$ efficiently.  
Note that this lemma does not generally hold for $\T$.

\begin{figure}[t]
	\begin{center}
		\includegraphics[scale=0.75]{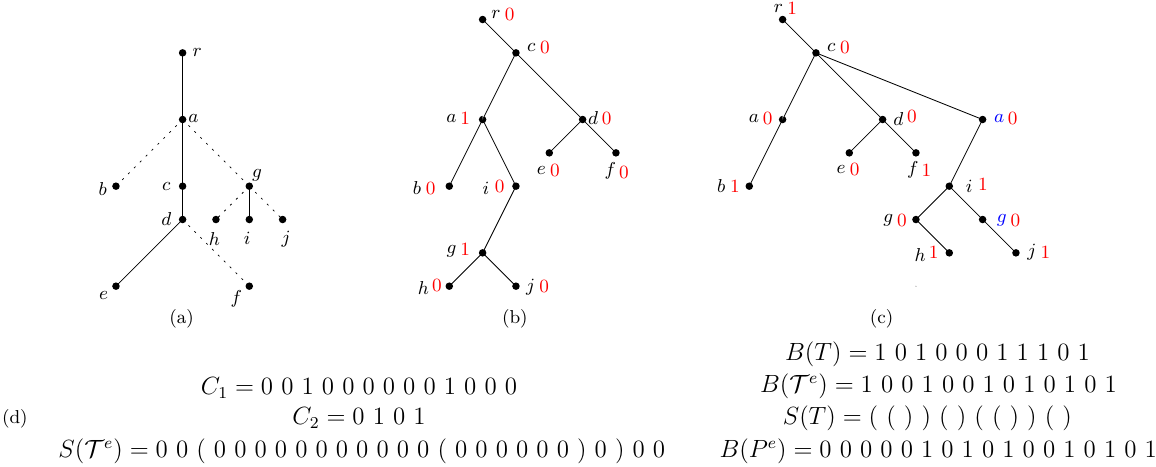}
	\end{center}
	\caption{(a) A tree $T$ (dotted lines indicate edges that are not part of any path in $P$), (b) a Steiner-closed search tree $\T$ on $T$, with red labels representing the labels of $\mathcal{L}$, (c) an extended search tree $\T^e$ (blue nodes indicate the copied nodes) of $\T$, with red labels representing the labels of $\mathcal{L}^e$, and (d) the sequences for supporting $\access{}$ queries on $\T$.}
	\label{figure:dummy}
\end{figure}

\begin{lemma}\label{lem:steiner_ds1}
	Let $b_i$ be the first traversed node in $P_i$ during the preorder traversal of $\T^e$. 
	If $s_i$ is the $j$-th node among $\{s_1, \dots, s_{\ell}\}$ in the preorder traversal of $T$, then $b_i$ is also the $j$-th node among $\{b_1, \dots, b_{\ell}\}$ in the preorder traversal of $\T^e$.
\end{lemma}
\begin{proof}
	From the definition of $\T^e$, $b_i$ is always an original node. 
	Suppose there exist $\alpha, \beta \in [\ell]$ such that $\preorder_T(s_{\alpha}) < \preorder_T(s_{\beta})$ but $\preorder_{\T^e}(b_{\alpha}) > \preorder_{\T^e}(b_{\beta})$. 
	We consider two cases: (1) $s_{\alpha}$ is an ancestor of $s_{\beta}$, and (2) $s_{\alpha}$ is not an ancestor of $s_{\beta}$ in $T$.
	
	In case (1), by Lemma~\ref{lem:steiner_left}, $b_{\alpha}$ lies on the leftmost path in $\T^e(b_{\beta})$. 
	We claim that all the nodes in $P_\T( b_{\beta}, b_{\alpha})$ lie on a single path in $T$. 
	Assume the opposite and let $b'$ be the node in $P_\T( b_{\beta}, b_{\alpha})$ with the minimum depth where the nodes in the path $P_\T( b_{\beta}, b')$ are not on a single path in $T$. 
	Since $P_\T( b_{\beta}, b_{\alpha})$ does not contain $\lca_T( b', b_{\alpha})$ (otherwise, $b'$ and $b_{\alpha}$ would not be on a single path in $\T$), the set of nodes in $P_\T( b_{\beta}, b_{\alpha})$ is not a Steiner-closed set, which is a contradiction. 
	Since the claim holds, from the construction of $\T$ in Section~\ref{sec:steiner}, $P_{\beta}$ would contain both $b_{\alpha}$ and $b_{\beta}$, leading to a contradiction.
	
	In case (2), let $u = \lca_T( s_{\alpha}, s_{\beta})$. 
	Then there exists a vertex $v \in P_T( u, s_{\alpha})$ where the leftmost child node $v_L$ of $v$ in $\T$ is an ancestor of $v$ in $T$. 
	If multiple such vertices exist, choose the one with the minimum depth. 
	Then $v$ is a proper descendant of $u$ in $T$; otherwise, $\preorder_{\T^e}(b_{\alpha}) < \preorder_{\T^e}(b_{\beta})$ from the construction of $\T$, which is a contradiction. 
	Moreover, $P_\T( r, v)$ does not contain $b_{\alpha}$; otherwise, the set of nodes in this path is not a Steiner-closed set since it would not contain $u$. 
	Thus, $v_L$ has a copied node in $\T^e$, and $b_{\beta}$ is in the rooted subtree of $\T^e$ at this copied node $v_L$. 
	This implies $\preorder_{\T^e}(b_{\alpha}) < \preorder_{\T^e}(b_{\beta})$, which is a contradiction.
\end{proof}

\begin{example}
	In the example of Figure~\ref{figure:dummy}, $s_1 = b$, $s_2 = r$, $s_3 = f$, $s_4 = h$, $s_5 = g$, and $s_6 = j$. Also, $b_1 = b$, $b_2 = r$, $b_3 = f$, $b_4 = h$, $b_5 = i$, and $b_6 = j$. During the preorder traversal of $T$, the nodes are visited in the order $s_2 \rightarrow s_1 \rightarrow s_3 \rightarrow s_5 \rightarrow s_4 \rightarrow s_6$. Also during the preorder traversal of $\T^e$, the nodes are visited in the order $b_2 \rightarrow b_1 \rightarrow b_3 \rightarrow b_5 \rightarrow b_4 \rightarrow b_6$. However, during the preorder traversal of $\T$, the nodes are visited in the order $b_2 \rightarrow b_1 \rightarrow b_5 \rightarrow b_4 \rightarrow b_6 \rightarrow b_3$.
\end{example}

Next, we describe an $O(n)$-bit data structure to compute $\preorder_{\T^e}(p)$ in $O(1)$ time when $\preorder_{\T}(p)$ is given (here, $p$ is an original node).  
Let $\mathcal{L}$ be a labeled tree obtained from $\T$ by replacing each node with a label of either $0$ or $1$ as follows:  
a node in $\mathcal{L}$ is labeled $1$ if and only if it has a copied node in $\T^e$.  
We store $\mathcal{L}$ in $O(n)$ bits to support tree navigational queries as in \cite{DBLP:journals/talg/NavarroS14},  
as well as labeled ancestor queries (i.e., finding the nearest ancestor of a given node with a specified label)  
in $O(1)$ time~\cite{DBLP:journals/algorithmica/HeMZ14}.  

Similarly, we define $\mathcal{L}^e$ as a labeled tree obtained from $\T^e$ by replacing each node with a label of either $0$ or $1$ as follows:  
a node in $\mathcal{L}^e$ is labeled $1$ if and only if it corresponds to a node in the set $\{b_1, \dots, b_{\ell}\}$,  where each $b_i$ is defined as in Lemma~\ref{lem:steiner_ds1}.  
We store $\mathcal{L}^e$ in $O(n)$ bits using the same data structures as for $\mathcal{L}$.  

In addition, we store the following auxiliary structures:

\begin{itemize}
	\item Let $C_1$ be a bit string of size $n$, where $C_1[i] = 1$ if and only if the node $u \in \T$ with $\preorder_{\T}(u) = i$ has a copied node in $\T^e$. 
	Now Suppose the node in $\T$ that corresponds to the $j$-th $1$ in $C_1$ has $d_j$ children where the parent node of them in $\T^e$ is an original node. From the construction of $\T^e$, these children are $d_j$ leftmost siblings in $\T$.
	We then define $C_2$ as a bit string of size at most $2n$, as $0^{d_1}10^{d_2}1 \dots 0^{d_t}1$ where $t$ is the total number of ones in $C_1$. We store $C_1$ and $C_2$ using the data structure of Lemma~\ref{lem:rankselect} using $O(n)$ bits to support rank and select queries on them in $O(1)$ time.
	
	\item Let $S(\T^e)$ be a sequence constructed from the BP of $\T^e$ where all the parenthesis are replaced by zero except (i) open parentheses corresponding to original nodes that have the copied ones, and (ii) closed parentheses corresponding to copied nodes.
	From the definition of $\T^e$, the subsequence of $S(\T^e)$ is a obtained by removing all zeros is balanced. Moreover, for any `(' in $S(\T^e)$ corresponding to an original node, its matching `)' corresponds to its copied node. 
	We maintain the $O(n)$-bit data structures of Lemma~\ref{lem:rankselect} and \ref{lem:bp} on $S(\T^e)$ to support rank, select, findopen, and findclose in $O(1)$ time. 
\end{itemize}

Using the above structures, we can compute $\preorder_{\T^e}(p)$ in $O(1)$ time when $\preorder_{\T}(p)$ is given, as follows:

\begin{enumerate}
	\item Find the nearest proper ancestor $p_1$ of $p$ in $\mathcal{L}$  
	with label $1$ using a labeled ancestor query.  
	If no such ancestor exists, then $\preorder_{\T}(p) = \preorder_{\T^e}(p)$.  
	Otherwise, find $p_2$, the child of $p_1$ that is an ancestor of $p$ in $\mathcal{L}$.  
	Both $p_1$ and $p_2$ can be located in $O(1)$ time using the tree navigational queries on $\mathcal{L}$.
	
	\item Using rank and select queries on $C_1$ and $C_2$,  
	together with tree navigational queries (degree and child-rank queries in Table~1 of \cite{DBLP:journals/talg/NavarroS14}) on $\mathcal{L}$,  
	we check whether $p_2$ is a child of the original or copied $p_1$ in $\T^e$ in $O(1)$ time.  
	If $p_2$ is a child of the original (or copied) $p_1$, we compute the preorder number of the original (or copied) $p_1$ in $\T^e$ in $O(1)$ time by:  
	(i) locating the corresponding open (or close) parenthesis in $S(\T^e)$ using the select operation on $C_1$ and the select and findclose operations on $S(\T^e)$, and  
	(ii) computing $\preorder_{\T^e}(p_1)$ using the data structure of \cite{DBLP:journals/talg/NavarroS14} on $\mathcal{L}^e$.  Using the preorder number of the original (or copied) $p_1$ in $\T^e$, we compute $\preorder_{\T^e}(p_2)$ using a child query (see Table~1 of \cite{DBLP:journals/talg/NavarroS14}) on $\mathcal{L}^e$.
	
	\item Let $\delta_p = \preorder_{\T}(p) - \preorder_{\T}(p_2)$.  
	Then $\preorder_{\T^e}(p) = \preorder_{\T^e}(p_2) + \delta_p$.
\end{enumerate}

\subparagraph*{Find the vertex $s_i$ in $T$. } 
Let $B(T)$ and $B(\T^e)$ be two bit strings of size $O(n)$ defined as follows:  
(i) $B(T)[i] = 1$ if and only if the vertex whose preorder number is $i$ in $T$ is in $\{s_1, \dots, s_{\ell}\}$, and (ii) $B(\T^e)[i] = 1$ if and only if the node whose preorder number is $i$ in $\T^e$ is in $\{b_1, \dots, b_{\ell}\}$.  We maintain both bit strings using the data structure of Lemma~\ref{lem:rankselect}, which requires $O(n)$ bits.

Then, we can find the vertex $s_i$ such that the path $P_i = P_T( s_i, l_i)$ contains the vertex $u_k$ in $O(1)$ time using rank and select on the bit strings together with tree navigational queries on  $\mathcal{L}^e$ as follows:  
(1) find the original node $u_k$ in $\T^e$ using its preorder number $k$ in $\T$;  
(2) find the node $b_i$ in $\T^e$, which is the nearest labeled ancestor of $u_k$ with label $1$ in $\mathcal{L}^e$, by Lemma~\ref{lem:steiner-closed3};  
(3) if $b_i$ corresponds to the $j$-th $1$ in $B(\T^e)$, then $\preorder_T(s_i)$ is the position of the $j$-th $1$ in $B(T)$ by Lemma~\ref{lem:steiner_ds1}.

\subparagraph*{Find the vertex $l_i$ in $T$. } 
Define a set $E = \{s_1, \dots, s_{\ell}\} \cup \{l_1, \dots, l_{\ell}\}$.  
We define a parenthesis sequence $S(T)$ of size $2\ell$ such that  
$S(T)[i]$ is `(' (resp., `)') if the $i$-th vertex in $E$, according to the preorder traversal of $T$,  
is in $\{s_1, \dots, s_{\ell}\}$ (resp., $\{l_1, \dots, l_{\ell}\}$).  
By Lemmas~\ref{lem:steiner-closed1} and \ref{lem:steiner-closed3}, $S(T)$ is balanced, and for each open parenthesis in $S(T)$ corresponding to $s_i$, its matching closed parenthesis corresponds to $l_i$.  

We store $S(T)$ using the data structures of Lemma~\ref{lem:rankselect} and Lemma~\ref{lem:bp},  
requiring $O(n)$ bits in total.  
Since we can compute the position in $S(T)$ corresponding to $s_i$ in $O(1)$ time using a rank operation on $B(T)$,  we can compute $\preorder_T(l_i)$ in $T$ in $O(1)$ time using findclose and rank operations on $S(T)$ together with the tree navigational query (\texttt{leaf\_select} query in Table~1 of \cite{DBLP:journals/talg/NavarroS14}) on $T$.

\subparagraph*{Computing $\inorder_{\T_i}(u_k)$. } 
We store the tree structures of $\T_1, \dots, \T_{\ell}$  
using the representation of Lemma~\ref{lem:path}, along with $o(n)$-bit auxiliary structures  
to support tree navigational queries on them~\cite{DBLP:journals/mics/DavoodiRS17}, which requires $O(n)$ bits in total. Thus, it suffices to show how to compute $\preorder_{\T_i}(u_k)$ in $O(1)$ time,  
since $\inorder_{\T_i}(u_k)$ can then be obtained from $\preorder_{\T_i}(u_k)$ in $O(1)$ time using the data structure of \cite{DBLP:journals/mics/DavoodiRS17}.
To compute $\preorder_{\T_i}(u_k)$,  
let $\Delta_{\T^e} = \preorder_{\T^e}(u_k) - \preorder_{\T^e}(b_i)$  
and $\Delta_{\T_i} = \preorder_{\T_i}(u_k) - \preorder_{\T_i}(b_i) = \preorder_{\T_i}(u_k) - 1$.  
Since some nodes in $\T^e(b_i)$ do not belong to $\T_i$,  
we generally have $\Delta_{\T^e} \ge \Delta_{\T_i}$.  
Because $\Delta_{\T^e}$ can be computed in $O(1)$ time as described earlier,  
we focus on how to compute $\delta = \Delta_{\T^e} - \Delta_{\T_i}$.  

Let $B(P^e)$ be a bit string of size $O(n)$,  
which is a unary encoding of the sizes of the trees $\T_1, \dots, \T_{\ell}$,  
ordered according to the preorder numbers of $b_1, \dots, b_{\ell}$ in $\T^e$.  
We store $B(P^e)$ using the data structure of Lemma~\ref{lem:rankselect},  
requiring $O(n)$ bits.  

To compute $\delta$, let $F$ be the set of nodes in $\{b_1, \dots, b_{\ell}\}$  
such that for any $b_j \in F$, $\preorder_{\T^e}(b_i) < \preorder_{\T^e}(b_j) < \preorder_{\T^e}(u_k)$. The size of $F$ can be computed in $O(1)$ time using rank and select queries on $B(\T^e)$. Then, by Lemmas~\ref{lem:steiner-closed1} and \ref{lem:steiner-closed3},  
we can compute $\delta$ by summing:  
(i) the sizes of the subtrees whose roots are in $F$, and  
(ii) the number of copied nodes in $\T^e$ that are traversed after $b_i$ and before $p$  
in the preorder traversal of $\T^e$.  
Both (i) and (ii) can be computed in $O(1)$ time using rank and select queries on $B(P^e)$ and $C_1$.  
\\\\
We summarize the result of the data structure in this section as the following theorem.

\begin{theorem}\label{thm:ds_steiner}
	Given a Steiner-closed search tree $\T$ on a tree $T$ with $n$ vertices,  
	there exists a data structure of size $O(n)$ bits that supports $\access$ query on $\T$ in $O(1)$ time.
\end{theorem}

\end{document}